\documentclass[11pt]{article}
\usepackage[skip=5.5pt]{parskip}
\usepackage[utf8]{inputenc}
\usepackage{graphicx}
\usepackage{dsfont}
\usepackage{amsmath, amssymb, amsthm}
\usepackage{mathtools}
\usepackage{physics}
\usepackage{diagbox}
\usepackage{tikz}
\usepackage[left=2.5cm,right=2.5cm,top=3cm,bottom=3cm]{geometry}
\usepackage{authblk}
\usepackage{pifont}
\usepackage{longtable}
\usepackage{caption}
\usepackage{pgfplots}
\pgfplotsset{compat=1.18}

\usepackage{libertinus}
\usepackage[T1]{fontenc}

\usepackage{tikz}
\usetikzlibrary{backgrounds,positioning,shapes.geometric,decorations.markings,arrows,knots,calc,decorations.pathmorphing,decorations.pathreplacing,calligraphy,calc,shapes,fit}
\tikzset{snake it/.style={
    decorate,
    decoration={snake,amplitude=0.07cm,post=lineto, post length=0.2cm},
}}
\tikzset{green_box/.style={
    very thick,
    text=mygreen!40!black,
    draw=mygreen!80!white,
    rounded corners,
    minimum height=1.1cm,
    minimum width=1.1cm,
    fill=mygreen!30,
}}
\tikzset{blue_box/.style={
    very thick,
    text=myblue!40!black,
    draw=myblue!80!white,
    rounded corners,
    minimum height=1.1cm,
    minimum width=1.1cm,
    fill=myblue!30,
}}

\usepackage[colorlinks=true,linkcolor=myblue,citecolor=myblue,urlcolor=myblue]{hyperref}
\usepackage[noabbrev,capitalize]{cleveref}
\usepackage{doi}

\definecolor{LightGray}{RGB}{220,220,220}
\definecolor{myred}{RGB}{255, 19, 0}
\definecolor{myblue}{RGB}{14, 81, 167}
\definecolor{myorange}{RGB}{255, 129, 0}
\definecolor{mygreen}{RGB}{0, 146, 44}
\definecolor{mypurple}{HTML}{9f53e0}
\definecolor{boxcolor}{HTML}{f1f1ff}
\definecolor{myyellow}{HTML}{ffd000}

\usepackage{tcolorbox}
\theoremstyle{definition}
\tcbuselibrary{breakable,skins}

\newtheorem{lemma}{Lemma}[section]
\tcolorboxenvironment{lemma}{
    colback=myblue!7!white,
    colframe=myblue!20!white,
}

\newtheorem{lemmaBreakable}[lemma]{Lemma}
\tcolorboxenvironment{lemmaBreakable}{
    colback=myblue!7!white,
    colframe=myblue!20!white,
    enhanced,
    breakable,
}

\newtheorem{proposition}[lemma]{Proposition}
\tcolorboxenvironment{proposition}{
    colback=myblue!7!white,
    colframe=myblue!20!white,
}

\newtheorem{theorem}[lemma]{Theorem}
\tcolorboxenvironment{theorem}{
    colback=myyellow!7!white,
    colframe=myyellow!25!white,
}

\newtheorem{corollary}[lemma]{Corollary}
\tcolorboxenvironment{corollary}{
    colback=myblue!7!white,
    colframe=myblue!20!white,
}

\theoremstyle{definition}

\newtheorem{definition}[lemma]{Definition}
\tcolorboxenvironment{definition}{
    colback=mygreen!6!white,
    colframe=mygreen!20!white,
    enhanced,
    breakable,
}

\newtheorem{remark}[lemma]{Remark}
\tcolorboxenvironment{remark}{
    colback=LightGray!20!white,
    colframe=LightGray!80!white,
    parbox=false,
    enhanced,
    breakable,
}

\allowdisplaybreaks

\title{\vspace{-2cm}Randomness from causally independent processes}
\author[1]{Martin Sandfuchs}
\author[1]{Carla Ferradini}
\author[1]{Renato Renner}

\affil[1]{Institute for Theoretical Physics, ETH Zurich, Zurich, Switzerland}
\date{}

\hypersetup{
    colorlinks,
    linkcolor=myblue,
    citecolor=myblue,
    filecolor=myblue,
    urlcolor=myblue,
    pdftitle={Randomness from causally independent processes},
    pdfauthor={Martin Sandfuchs, Carla Ferradini and Renato Renner},
}

\newcommand{\Ssub}[1]{\mathcal{S}_\bullet(#1)}

\newcommand{\cM}{\mathcal{M}}
\newcommand{\cN}{\mathcal{N}}

\begin{document}

\maketitle

\begin{abstract}
    We consider a pair of causally independent processes, modelled as the tensor product of two channels, acting on a possibly correlated input to produce random outputs $X$ and $Y$. We show that, assuming the processes produce a sufficient amount of randomness, one can extract uniform randomness from $X$ and $Y$. This generalizes prior results, which assumed that $X$ and $Y$ are (conditionally) independent. Note that in contrast to the independence of quantum states, the independence of channels can be enforced through spacelike separation. As a consequence, our results allow for the generation of randomness under more practical and physically justifiable assumptions than previously possible. We illustrate this with the example of device-independent randomness amplification, where we can remove the constraint that the adversary only has access to classical side information about the source.
\end{abstract}

\section{Introduction} \label{sec:introduction}

Consider the following scenario which is shown in \cref{fig:setup_spacetime}. Two experimentalists are located in two distant places, say Zurich and Sydney. Simultaneously, they both perform experiments designed to generate randomness, $X$ and $Y$, respectively.\footnotemark{} Due to their geographic locations, $X$ and $Y$ are produced in a spacelike separated fashion, i.e., there is no causal influence from Zurich to Sydney or vice versa during the course of the experiment. However, because of experimental imperfections, neither $X$ or $Y$ are perfectly random. Furthermore, the two experimentalists' data may be correlated due to the influence of events in their common past (e.g., solar activity). Nevertheless, since $X$ and $Y$ were produced by independent processes (enforced by the spacelike separation), they cannot be too badly correlated. As a result, we may hope to construct a function $\mathrm{Ext}$ such that $Z = \mathrm{Ext}(X, Y)$ is a string of uncorrelated bits. A diagram of the model considered in this paper is given in \cref{fig:setup_circuit} below.

\begin{figure}[ht!]
    \centering
    \begin{tikzpicture}[scale=0.85]
        \node (X) at (-3, 0) {};
        \node (Y) at (+3, 0) {};

        \begin{scope}[on background layer={}]
            \fill[fill=myblue!10!white,opacity=0.7] (-6, -3) arc(-180:0:+3cm and 0.65cm) -- (X.center) -- cycle;
            \fill[fill=myblue!10!white,opacity=0.7] (+6, -3) arc(-180:0:-3cm and 0.65cm) -- (Y.center) -- cycle;
        \end{scope}

        \draw[fill=mygreen!30] (-5.5, -2.5) arc(-180:0:2.5cm and 0.4cm) -- (X.center) -- cycle;
        \node[text=mygreen!40!black] at (-3, -1.3) {$\cM$};

        \fill[very thick,fill=myorange!50!white,opacity=0.5] (-3, -2.5) ellipse (2.5cm and 0.4cm);
        \draw[very thick,myorange!70!black,dotted] (-5.5, -2.5) arc(+180:0:2.5cm and 0.4cm);
        \node[text=myorange!70!black] at (-3, -2.5) {$A$};

        \draw[very thick,draw=mygreen!80!white] (-5.5, -2.5) arc(-180:0:2.5cm and 0.4cm) -- (X.center) -- cycle;

        \draw[fill=mygreen!30] (+5.5, -2.5) arc(-180:0:-2.5cm and 0.4cm) -- (Y.center) -- cycle;
        \node[text=mygreen!40!black] at (+3, -1.3) {$\cN$};

        \fill[very thick,fill=myorange!50!white,opacity=0.5] (+3, -2.5) ellipse (-2.5cm and 0.4cm);
        \draw[very thick,myorange!70!black,dotted] (+5.5, -2.5) arc(+180:0:-2.5cm and 0.4cm);
        \node[text=myorange!70!black] at (+3, -2.5) {$B$};

        \draw[very thick,draw=mygreen!80!white] (+5.5, -2.5) arc(-180:0:-2.5cm and 0.4cm) -- (Y.center) -- cycle;

        \filldraw[black] (X) circle (2pt) node[above=0.1cm]{$X$};
        \filldraw[black] (Y) circle (2pt) node[above=0.1cm]{$Y$};

        \draw[very thick,->,>=stealth] (-7.0, -4.0) -- (-7.0, +0.5) node[above]{$t$};
        \draw[very thick,->,>=stealth] (-7.0, -4.0) -- (+7.5, -4.0) node[right]{$x$};

        \begin{scope}[on background layer={}]
            \draw[very thick,dashed,draw=gray] (-7.2, 0) node[left,black]{$t_1$} -- (+7.2, 0);
        \draw[very thick,dashed,draw=gray] (-7.2, -2.5) node[left,black]{$t_0$} -- (+7.2, -2.5);
        \end{scope}

        \draw[very thick] (-3, -3.9) -- (-3, -4.1) node[below]{Zurich};
        \draw[very thick] (+3, -3.9) -- (+3, -4.1) node[below]{Sydney};
    \end{tikzpicture}
    \caption{
        \textbf{Spacetime diagram illustrating the generation of $X$ and $Y$.} Two randomness generating processes begin at time $t_0$ and finish producing randomness by time $t_1$. Due to the spatial distance between the two experimentalists, the two processes $\cM$ and $\cN$ act independently on $A$ and $B$, which are spacelike separated regions of the Cauchy surface at time $t_0$.
    }
    \label{fig:setup_spacetime}
\end{figure}
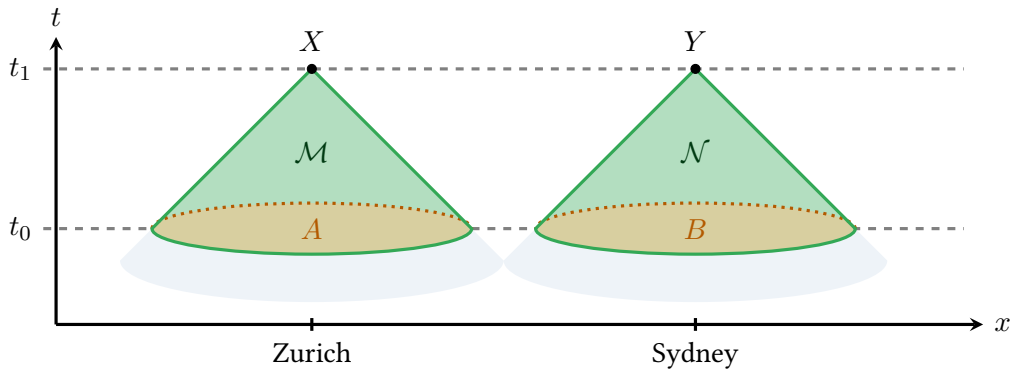
\footnotetext{For concreteness, one can imagine that they both perform (imperfect) polarization measurements on suitably prepared photons (see, for instance, \cite{Frauchiger_2013}). Note that, in contrast to classical processes, measurement results in quantum mechanics can be fundamentally unpredictable \cite{Heisenberg_1927, Bell_1964, Berta_2010}.}

The function $\mathrm{Ext}$ described above is commonly referred to as a two-source extractor and has been studied extensively in both classical and quantum information theory (see \cite{Chattopadhyay_2022} for a review of the classical literature). Initially, researchers considered the scenario when $X$ and $Y$ are independent random variables \cite{Santha_1986,Chor_1988}. This has since been extended to the situation where the adversary holds quantum side information. Specifically, in \cite{Kasher_2010}, the authors considered states of the form $\rho_{XYC_1 C_2} = \rho_{X C_1} \otimes \rho_{Y C_2}$, i.e., the side information about $X$ is independent from the side information about $Y$.\footnote{In \cite{Kasher_2010}, they also consider adversaries holding entangled side information. However, they only obtain results against adversaries with bounded quantum storage, an assumption we don't make here.} In \cite{Friedman_2016}, this was generalized to states $\rho_{XYC}$ satisfying the Markov chain condition $X \leftrightarrow C \leftrightarrow Y$, which can be interpreted as $X$ and $Y$ being independent when conditioned on $C$ \cite{Hayden_2004}. It is easy to see that if $\rho_{A B}$ in \cref{fig:setup_circuit} is a purely classical (or, more generally, separable) state, then one obtains that $X$ and $Y$ are independent when conditioned on the channel inputs $A$ and $B$, i.e., $X \leftrightarrow A B \leftrightarrow Y$ forms a classical Markov chain. Hence, for classically correlated inputs, our setup can be treated using the Markov model considered in \cite{Friedman_2016} (see also the discussion in \cref{subsec:markov_model}). This validates our intuition that a state produced by two independent processes is sufficiently uncorrelated to extract randomness. 
However, for entangled inputs, our model can no longer be captured by quantum Markov chains (we formally show this in \cref{thm:non_markovianity}). In this sense, the setup in \cref{fig:setup_circuit} can be seen as a generalized notion of conditional independence beyond quantum Markov chains.

\begin{figure}[ht!]
    \centering
    \begin{tikzpicture}
        \node[green_box] (M) at (0, +0.8) {$\cM$};
        \node[green_box] (N) at (0, -0.8) {$\cN$};
        \node at (0, 0) {$\otimes$};
        \draw[thick,->,>=stealth,dashed] ([xshift=-0.8cm] M.west) -- node[above]{$A$} (M.west);
        \draw[thick,->,>=stealth,dashed] ([xshift=-0.8cm] N.west) -- node[above]{$B$} (N.west);

        \node[blue_box,minimum height=2.7cm,minimum width=1.3cm] (Ext) at (2.3, 0) {$\mathrm{Ext}$};
        \draw[thick,->,>=stealth] (M.east) -- node[above]{$X$} (M.east-|Ext.west);
        \draw[thick,->,>=stealth] (N.east) -- node[above]{$Y$} (N.east-|Ext.west);
        \node[] (Z) at ([xshift=1.1cm] Ext.east) {};
        \draw[thick,->,>=stealth] (Ext.east) -- node[above]{$Z$} (Z.west);

        \draw[thick,->,>=stealth,dashed,rounded corners=0.1cm] (M.north) -- ([yshift=+0.5cm] M.north) -- node[above]{$S$} ([yshift=+0.5cm] M.north-|Z.west);
        \draw[thick,->,>=stealth,dashed,rounded corners=0.1cm] (N.south) -- ([yshift=-0.5cm] N.south) -- node[below]{$T$} ([yshift=-0.5cm] N.south-|Z.west);

        \draw[very thick,decorate,decoration={calligraphic brace,amplitude=0.15cm,raise=0.1cm}] ([xshift=-0.8cm,yshift=-0.2cm] N.west) -- node[left=0.25cm]{$\rho_{AB}$} ([xshift=-0.8cm,yshift=+0.2cm] M.west);

        \draw[very thick,decorate,decoration={calligraphic brace,amplitude=0.15cm,raise=0.0cm}] ([yshift=+0.6cm] M.north-|Z) -- node[right=0.2cm]{$\rho_{ZST}^\mathrm{Ext}$} ([yshift=-0.6cm] N.south-|Z);
    \end{tikzpicture}
    \caption{
        \textbf{Circuit diagram of the setup.} Two independent channels $\cM$ and $\cN$ are applied to an initial quantum state $\rho_{AB}$ to produce classical values $X$ and $Y$, respectively. Additionally, we allow the channels to produce quantum side information $S$ and $T$. The state $\rho_{AB}$ should be understood to capture all degrees of freedom that $\cM$ and $\cN$ may depend on (see also \cref{fig:setup_spacetime}). An extractor function $\text{Ext}$ produces a random bitstring~$Z$, which should be uniformly distributed and independent from~$S$ and~$T$. The length of the generated bitstring~$Z$ depends on the amount of randomness---measured in terms of entropy---produced by the channels $\cM$ and $\cN$. Note that one may also consider an extra purifying system $E$ for $\rho_{AB}$. This could be passed through $\cM$ or $\cN$, i.e., there is no need to explicitly model the identity channel on $E$.
    }
    \label{fig:setup_circuit}
\end{figure}
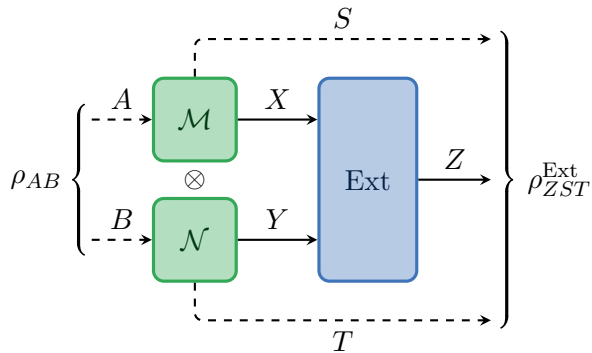

In practice, it is hard (or even impossible) to justify the (conditional) independence of the state of two systems: even if they are spatially separated, they could depend on a common past. On the other hand, as illustrated by our introductory example, the causal independence of quantum processes can be experimentally enforced.\footnote{Causal independence holds even without spacelike separation if the processes take place in separate and sealed laboratories, as is commonly assumed in cryptography.} This makes our setup attractive for constructing quantum random number generators, where one aims to eliminate unnecessary device assumptions. Apart from being easier to justify, our model also allows for new applications. As an example of this, in \cref{sec:quantum_DIRA}, we demonstrate how our results can be used to prove the security of device-independent randomness amplification schemes when the adversary holds quantum side information about the source of randomness (as opposed to the classical side information considered in, for example, \cite{Kessler_2020, Foreman_2023}).

The remainder of this paper is organized as follows. In \cref{sec:preliminaries}, we summarize some preliminaries. Readers familiar with the formalism of quantum information theory should feel free to skip this section. In \cref{sec:two_process_extractors} we formally introduce our model of extractors, the \emph{two-process extractors}, which will be the object of interest throughout the remaining sections. In \cref{sec:ip_security}, we show that a simple construction, the inner product construction, can be used to extract a single bit of uniform randomness in our model. In \cref{sec:deor_security}, we extend these results to extract multiple bits of randomness. Next, in \cref{sec:smoothing}, we show that our model is robust, i.e., the extractors still work when the entropy conditions are only satisfied approximately. In \cref{sec:prior_work}, we discuss the relation of our model to prior work. In \cref{sec:quantum_DIRA}, we apply our results to device-independent randomness amplification protocols. Finally, in \cref{sec:conclusion}, we summarise the main conclusions and discuss some open problems.

\section{Preliminaries and notation} \label{sec:preliminaries}

Here, we summarize some of the main notations and quantities used in the statements and proofs that follow. For a detailed introduction to the formalism of quantum information theory, we refer to the literature, e.g., \cite{Nielsen_2010}. Note that, somewhat unconventionally, throughout this paper we will allow for sub-normalized states and channels. That is, when we say ``state'', we mean a positive semi-definite linear operator $\rho$ with $\tr[\rho] \leq 1$. A summary of the notation is given in \cref{tab:notation} below.

\renewcommand{\arraystretch}{1.4}
\begin{longtable}[H]{|c||p{13cm}|}
        \hline
        \textbf{Notation} & \textbf{Description} \\
        \hline
        \hline
        $A^n$ & The composite system $A_1 \ldots A_n$ \\ \hline
        $\mathrm{Lin}(A, B)$ & Set of linear operators from the space $A$ to $B$ \\ \hline
        $\mathrm{Lin}(A)$ & The same as $\mathrm{Lin}(A, A)$ \\ \hline
        $L_{B|A}^*$ & The adjoint of $L_{B|A} \in \mathrm{Lin}(A, B)$ \\ \hline
        $S \perp T$ & $S$ and $T$ are orthogonal, i.e., $ST = TS = 0$ \\ \hline
        $S \geq 0$ & $S\in\mathrm{Lin}(A)$ is positive semi-definite \\ \hline
        $S \leq T$ & $T - S \geq 0$, i.e., $T - S$ is positive semi-definite \\ \hline
        $\Ssub{A}$ & The set of sub-normalized density operators on system $A$, i.e., $\Ssub{A} = \{\rho_{A} \in \mathrm{Lin}(A): \rho_{A} \geq 0, 0 < \tr[\rho] \leq 1\}$ \\ \hline
        $\rho_{AB}$ & Density operator acting jointly on systems $A$ and $B$, i.e., $\rho_{AB} \in \Ssub{AB}$ \\ \hline
        $\rho_A$ & Reduced density operator on $A$ obtained by tracing out $B$: $\rho_A = \tr_B[\rho_{AB}]$ \\ \hline
        $\rho_X$ & Classical state on $\mathcal{H}_X$, describing a random variable $X$ with alphabet~$\mathcal{X}$: $\rho_X=\sum_{x \in \mathcal{X}} P_X(x) \ketbra{x}_X$, for a fixed computational basis $\{\ket{x}\}_{x}$ of $\mathcal{H}_X$  \\ \hline
        $\rho_{XA}$ & Classical-quantum state describing a random variable $X$ correlated with a quantum system $A$: $\rho_{XA}=\sum_{x \in \mathcal{X}} \ketbra{x}_X \otimes \rho_{A\land X=x}$ \\ \hline
        $\omega_Z$ & The maximally mixed state on the system $Z$ \\ \hline
        $S_{AB} T_{BC}$ & Shorthand for $(S_{A B} \otimes \mathds{1}_C)(\mathds{1}_A \otimes T_{BC})$ \\ \hline
        $\mathcal{I}_{R}$ & Identity channel on the system $R$ \\ \hline
        $\mathcal{E}_{B|A}$ & A channel, i.e., a completely-positive and trace non-increasing (CPTNI) map from $\mathrm{Lin}(A)$ to $\mathrm{Lin}(B)$ \\ \hline
        $\mathcal{E}_{B|A}[\rho_{AR}]$ & Application of a channel to a state of a larger system, i.e., $\mathcal{E}_{B|A}[\rho_{AR}] \coloneqq (\mathcal{E}_{B|A} \otimes \mathcal{I}_R)[\rho_{AR}]$ \\ \hline
        $f_{Y|X}[\rho_{XR}]$ & Notation for $f_{Y|X}[\rho_{XR}] = \sum_{x} \ketbra{f(x)}_Y \mel{x}{\rho_{XR}}{x}_X$, i.e., the function $f$ applied as a channel \\ \hline
        $f_{YX|X}[\rho_{XR}]$ & The same as $f_{Y|X}[\rho_{XR}]$ but with a copy of $X$ appended to the output. Explicitly, $f_{YX|X}[\rho_{XR}] = \sum_x \ketbra{f(x)}_Y \otimes \ketbra{x}_X \otimes \mel{x}{\rho_{XR}}{x}_X$ \\ \hline
        $x \cdot y$ & The inner product between $x$ and $y$, i.e., $x \cdot y = \sum_i x_i y_i$. \\ \hline
        $\ket{\Omega}_{AA'}$ & The non-normalized maximally entangled state, i.e., $\ket{\Omega}_{AA'} = \sum_i \ket{i}_{A} \otimes \ket{i}_{A'}$ \\ \hline
        $\Omega_{AA'}$ & The non-normalized state $\Omega_{AA'} = \ketbra{\Omega}_{AA'}$ \\ \hline
        $\|S\|_1$ & Schatten $1$-norm of $S$, given by $\|S\|_1 = \tr_{}[\sqrt{S^* S}]$ \\ \hline
        $\log$ & Logarithm to the base $2$ \\ \hline
        
    \caption{
        \textbf{Summary of notation.} Subscripts in capital letters refer to systems. We use $A, B, \ldots$ for generic quantum systems, while $X, Y, Z$ refer to classical systems, i.e., systems whose states are diagonal in a fixed computational basis.
    }
    \label{tab:notation}
\end{longtable}

\begin{remark}
    Note that if $\rho_A$ is a state, then $\rho_A^{1/2} \Omega_{AA'} \rho_A^{1/2}$ is a purification of $\rho_A$. Furthermore, for any $K_A \in \mathrm{Lin}(A)$ it holds that $K_A \ket{\Omega}_{AA'} = K_{A'}^T \ket{\Omega}_{AA'}$, where $\circ^T$ denotes the transpose in the basis underlying the definition of $\ket{\Omega}_{AA'}$. One can easily show that $\smash{(\rho_A^T)^{1/2}} = \smash{(\rho_A^{1/2})^T}$ and, similarly, $(\rho_{A}^{-1})^T = (\rho_A^T)^{-1}$.
\end{remark}

\begin{definition}[Instruments] \label{def:channel}
    An \emph{instrument} is a channel $\cM_{XS|A}$ where $X$ is a classical system. Any instrument can be decomposed as
    \begin{equation}
        \cM_{XS|A}[S_A] = \sum_x \ketbra{x}_X \otimes \cM^{x}_{S|A}[S_A],
    \end{equation}
    for some CPTNI maps $\cM^{x}_{S|A}$.
\end{definition}

\begin{definition}[Adjoint channel]
    For any channel $\mathcal{E}_{B|A}$, we denote by $\mathcal{E}_{B|A}^*$ its adjoint with respect to the Hilbert-Schmidt inner product, i.e., the unique superoperator such that $\tr_B[T_B^* \mathcal{E}_{B|A}[S_A]] = \tr_A[(\mathcal{E}_{B|A}^*[T_B])^* S_A]$ holds for all $S_A \in \mathrm{Lin}(A)$ and $T_B \in \mathrm{Lin}(B)$. Note that if $\mathcal{E}_{B|A}$ is completely positive, then so is $\mathcal{E}_{B|A}^*$. If $\mathcal{E}_{B|A}$ is trace non-increasing, then $\mathcal{E}_{B|A}^*$ is sub-unital, i.e., $\mathcal{E}_{B|A}^*[\mathds{1}_B]\leq \mathds{1}_A$.
\end{definition}

\begin{lemma}[Stinespring dilation \cite{Stinespring_1955}] \label{lem:stinespring}
    Let $\mathcal{E}_{B|A}$ be a channel. Then, there exists $K_{BR|A} \in \mathrm{Lin}(A,BR)$, called a \emph{Stinespring dilation}, such that
    \begin{equation}
        \mathcal{E}_{B|A}[S_A] = \tr_{R}\left[K_{BR|A} S_A K_{BR|A}^*\right].
    \end{equation}
    Furthermore $K_{BR|A}^* K_{BR|A} \leq \mathds{1}_A$ with equality iff $\mathcal{E}_{B|A}$ is trace-preserving.
\end{lemma}
To quantify the quality of randomness, we will require some measure of distance. Since we will be dealing with sub-normalized states, some care is required when defining our distance measures. 

\begin{definition}[Trace norm]
  Let $S$ be a linear operator. Define the \emph{trace norm} by
  \begin{equation}
    \norm{S}_+ \coloneqq \max_{0 \leq \Lambda \leq \mathds{1}} \left| \tr[\Lambda S] \right|.
  \end{equation}
\end{definition}

\begin{remark}[Relation to $1$-norm]
    If $\rho$ and $\sigma$ are positive operators then \cite[Section 3.2]{Tomamichel_2016}
    \begin{equation}
        \norm{\rho - \sigma}_+ = \frac{1}{2} \norm{\rho - \sigma}_1 + \frac{1}{2} \abs{\tr[\rho] - \tr[\sigma]}.
    \end{equation}
    In particular, for states such that $\tr[\rho] = \tr[\sigma]$ we have that $\norm{\rho - \sigma}_+ = \frac{1}{2} \norm{\rho - \sigma}_1$. More generally, the equality above implies
    \begin{equation}
        \frac{1}{2} \|\rho - \sigma\|_1 \leq \|\rho - \sigma\|_+ \leq \|\rho - \sigma\|_1.
    \end{equation}
\end{remark}

For technical reasons, the following distance measure will prove to be useful.

\begin{definition}[Purified distance] \label{def:purified_distance}
    Let $\rho_{A},\sigma_{A} \in \Ssub{A}$. Define the \emph{purified distance} by
    \begin{equation}
        P\left( \rho_{A},\sigma_{A} \right) \coloneqq \inf_{\rho_{AB},\sigma_{AB}}\norm{\rho_{AB} - \sigma_{AB}}_{+},
    \end{equation}
    where the infimum runs over all purifications $\rho_{AB}$ and $\sigma_{AB}$ of $\rho_{A}$ and $\sigma_{A}$, respectively.
\end{definition}
\begin{remark}
    By the data-processing inequality for $\norm{\circ}_{+}$ we have that $\norm{\rho - \sigma}_{+} \leq P(\rho,\sigma).$
\end{remark}

The following property of the purified distance will be useful.

\begin{lemma}[{\cite[Corollary 9]{Tomamichel_2010}}] \label{lem:purified_extension}
    Let $\rho_{AB} \in \Ssub{AB}$ and $\sigma_{A} \in \Ssub{A}$. Then, there exists an extension $\sigma_{AB} \in \Ssub{AB}$ of $\sigma_{A}$ such that $P(\rho_{AB}, \sigma_{AB}) = P(\rho_A, \sigma_A)$.
\end{lemma}
To quantify the amount of randomness in the outputs $X$ and $Y$, we will use the following entropic quantities.

\begin{definition}[{Rényi entropies \cite{Muller_Lennert_2013,Wilde_2014}}] \label{def:renyi_divergence}
    Let $\alpha \in \left[ \frac{1}{2},\infty \right]$, $\rho \in \Ssub{A}$ and $\sigma \geq 0$. Define the \emph{sandwiched Rényi divergence} of order $\alpha$ as
    \begin{equation}
    \begin{aligned}
        D_{\alpha}(\rho,\sigma) \coloneqq & \begin{cases}
            \frac{1}{\alpha - 1}\log\left( \tr[\left( \sigma^{\frac{1 - \alpha}{2\alpha}}\rho\sigma^{\frac{1 - \alpha}{2\alpha}} \right)^{\alpha}] \right) & \text{if $(\alpha < 1 \land \rho \not\perp \sigma)$ or $(\mathrm{supp}(\rho) \subseteq \mathrm{supp}(\sigma))$} \\
            +\infty & \text{otherwise}
        \end{cases}.
    \end{aligned}
    \end{equation}
    Let $\rho_{AB} \in \Ssub{AB}$. Define the \emph{sandwiched conditional Rényi entropy}
    \begin{equation}
    \begin{aligned}
        H_{\alpha}^{\downarrow}\left( A|B \right)_{\rho} \coloneqq & - D_{\alpha}\left( \rho_{AB},\mathds{1}_{A} \otimes \rho_{B} \right) \\
        H_{\alpha}^{\uparrow}\left( A|B \right)_{\rho} \coloneqq & \max_{\sigma_{B}} - D_{\alpha}\left( \rho_{AB},\mathds{1}_{A} \otimes \sigma_{B} \right).
    \end{aligned}
    \end{equation}
    We also use the standard notation
    $H_{\min} \coloneqq H_{\infty}^{\uparrow}$.
\end{definition}
\begin{remark}
    In \cref{def:renyi_divergence} we use the convention from \cite{Wilde_2014} without the normalization by $\tr[\rho]$ as is done in \cite{Muller_Lennert_2013,Tomamichel_2016}. Note that this has no impact on the definition of $H_{\min}$.
\end{remark}

\begin{definition}[Smooth min-entropy]
    Let $\rho_{AB} \in \Ssub{AB}$ and $0 \leq \varepsilon < \sqrt{\tr[\rho]}$. The \emph{conditional smooth min-entropy} of $A$ given $B$ is defined by
    \begin{equation}
        H_{\min}^{\varepsilon} (A|B)_{\rho} \coloneqq \sup_{\tilde{\rho} \in \mathcal{B}_{\rho}^{\varepsilon}}H_{\min}(A|B)_{\tilde{\rho}}.
    \end{equation}
    Similarly, we define
    \begin{equation}
        H_{\min}^{\downarrow,\epsilon}(A|B)_\rho \coloneqq \sup_{\tilde{\rho} \in \mathcal{B}_\rho^\varepsilon} H_{\infty}^\downarrow(A|B)_{\tilde{\rho}}.
    \end{equation}
    In both definitions we use $\mathcal{B}_\rho^\varepsilon \coloneqq \{\tilde{\rho}_{AB} \in \Ssub{AB}: P(\tilde{\rho}_{AB},\rho_{AB}) \leq \varepsilon \}$.
\end{definition}

\section{Two-process extractors} \label{sec:two_process_extractors}

As explained in the introduction, the objective is to use $X$ and $Y$ to produce an almost uniformly random bitstring $Z$. Naturally, for this one needs a measure for how close $Z$ is to a perfectly random bitstring. We will characterize the quality of $Z$ in terms of the trace distance, as is commonly done in cryptography \cite{Ben-Or_2005, Renner_2006, Portmann_2022}. Let us introduce the following terminology.
\begin{definition}
    Let $\rho_{XYA}$ be a quantum state where $X$ and $Y$ are classical. Given some function $\mathrm{Ext}: \mathcal{X} \times \mathcal{Y} \rightarrow \mathcal{Z}$, we say that $Z = \mathrm{Ext}(X, Y)$ is \emph{$\varepsilon$-random relative to $A$} if
    \begin{equation}
        \frac{1}{2} \norm{\mathrm{Ext}_{Z|XY}[\rho_{XYA}] - \omega_{Z} \otimes \rho_{A}}_1 \leq \varepsilon,
    \end{equation}
    where $\omega_Z$ is the maximally mixed state on $Z$. Similarly, we say that $Z = \mathrm{Ext}(X, Y)$ is \emph{$\varepsilon$-random relative to $YA$} if
    \begin{equation}
        \frac{1}{2} \norm{\mathrm{Ext}_{ZY|XY}[\rho_{XYA}] - \omega_{Z} \otimes \rho_{YA}}_1 \leq \varepsilon.
    \end{equation}
\end{definition}
The above definition can be understood as requiring that $\rho_{ZA}$ behaves as $\omega_Z \otimes \rho_A$ except with probability $\varepsilon$ \cite{Ferradini_2025}.

As stated in the introduction, our goal is to find a function $\mathrm{Ext}$ such that $Z = \mathrm{Ext}(X, Y)$ is $\varepsilon$-random whenever $X$ and $Y$ were produced by causally independent and sufficiently random processes (see \cref{fig:setup_circuit}). This motivates the following definition.
\begin{definition}[Two-process extractor] \label{def:two_process_extractor}
    Let $k_1, k_2, \varepsilon \geq 0$.
    We call a function $\mathrm{Ext}: \{0,1\}^{n_1} \times \{0,1\}^{n_2} \rightarrow \{0,1\}^m$ a \emph{$(k_1, k_2, \varepsilon)$-weak two-process extractor}
    if for all pure states $\rho_{AB}$ and all instruments $\cM_{XS|A}$ and $\cN_{YT|B}$ with 
    \begin{equation}
        H_{\min}(X|SB)_{\cM[\rho]} \geq k_1 \quad \text{and} \quad H_{\min}(Y|TA)_{\cN[\rho]} \geq k_2,
    \end{equation}
    the state $\rho_{XYST}^{\mathrm{out}} \coloneqq \left(\cM_{XS|A} \otimes \cN_{YT|B}\right)[\rho_{AB}]$ is such that $Z = \mathrm{Ext}(X,Y)$ is $\varepsilon$-random relative to $ST$.
    \vspace{5.5pt}

    Similarly, we call $\mathrm{Ext}$ a \emph{$(k_1, k_2, \varepsilon)$ two-process extractor strong in $Y$}, if for all instruments and states as above with 
    \begin{equation}
        H_{\min}(X|SB)_{\cM[\rho]} \geq k_1 \quad \text{and} \quad H_{\min}(Y|A)_{\cN[\rho]} \geq k_2,
    \end{equation}
    the state $\rho_{XYST}^{\mathrm{out}}$ 
    is such that $Z=\mathrm{Ext}(X,Y)$ is $\varepsilon$-random relative to $YST$.
\end{definition}

\begin{remark}[Purity of input state] \label{rem:non_pure_inputs}
    \Cref{def:two_process_extractor} requires the input state $\rho_{AB}$ to be pure. This is mostly for convenience of notation. One can easily apply \cref{def:two_process_extractor} to non-pure $\rho_{AB}$. For this, let $\rho_{AB}$ be an arbitrary density operator with purification $\rho_{ABR}$. Let us define $\rho^\mathrm{Ext}_{ZST} = (\mathrm{Ext}_{Z|XY} \circ \cM_{XS|A} \otimes \cN_{YT|B} \otimes \tr_R)[\rho_{ABR}]$.
    We can then apply \cref{def:two_process_extractor} to $\rho_{ABR}$, $\cM_{XS|A}$ and $\cN_{YT|B} \otimes \tr_R$ to bound
    \begin{equation}
        \frac{1}{2} \norm{\rho^{\mathrm{Ext}}_{ZST} - \omega_Z \otimes \rho^\mathrm{Ext}_{ST}}_1 \leq \varepsilon.
    \end{equation}
    Note, however, that the entropy conditions now need to be applied to the purification $\rho_{ABR}$. More precisely, they now read
    \begin{equation}
        H_{\min}(X|SBR)_{\cM[\rho]} \geq k_1 \quad \text{and} \quad H_{\min}(Y|TA)_{(\cN \otimes \tr)[\rho]} \geq k_2
    \end{equation}
    for weak extractors and
    \begin{equation}
        H_{\min}(X|SBR)_{\cM[\rho]} \geq k_1 \quad \text{and} \quad H_{\min}(Y|A)_{(\cN \otimes \tr)[\rho]} \geq k_2
    \end{equation}
    for strong extractors. The above conditions can be understood as requiring that $\cM$ produces new entropy instead of simply passing along the entropy already contained in $\rho_{AB}$.
    
    Above, we decided to apply \cref{def:two_process_extractor} to the channels $\cM_{XS|A}$ and $\cN_{YT|B} \otimes \tr_R$. Alternatively, one could also use the channels $\cM_{XS|A} \otimes \tr_R$ and $\cN_{YT|B}$ to swap the $R$ system between the two entropies.
\end{remark}

\begin{remark}[Alternative model for randomness extraction] \label{rem:alt_model}
    In \cref{sec:alternative_model}, we consider a different model in which only $Y$ is produced by applying the instrument $\cN$, whereas $X$ is already part of the initial state. For some applications, such as device-independent randomness amplification considered in \cref{sec:quantum_DIRA}, this model can be more convenient. We show that this model is equivalent to the notion of two-process extractors given above.
\end{remark}

\section{Extracting a single bit}
\label{sec:ip_security}

A well-known extractor for independent $X$ and $Y$ is the inner product construction \cite{Vazirani_1985,Chor_1988}. We will first define the inner product construction and then show that it can also be used to extract randomness in our model.

\begin{definition}[Inner product (IP) construction]
  Let $x$ and $y$ be bitstrings of length $n$. Define the \emph{inner product construction} $\mathrm{IP}^n: \{0, 1\}^n \times \{0, 1\}^n \rightarrow \{0, 1\}$ by
  \begin{equation}
    \mathrm{IP}^n(x, y) \coloneqq x \cdot y = \sum_i x_i y_i,
  \end{equation}
  where addition is modulo 2.
\end{definition}

The following lemma shows that the inner product construction can be used to extract randomness in a slightly different setup from what is considered in \cref{def:two_process_extractor}. More precisely, it considers the scenario where only $Y$ is produced by an instrument $\cN_{YT|B}$ whereas $X$ is already part of the input state $\rho_{XB}$ (see also \cref{sec:alternative_model} and \cref{rem:alt_model}).

\begin{lemma} \label{lem:ip_ext_strong}
    Let $\rho_{XB}$ be a cq state and $\cN_{YT|B}$ be an instrument. Define $\rho_{XYT}^{\text{out}} \coloneqq \cN_{YT|B}\left[ \rho_{XB} \right]$,
    then, for any $\sigma_{B} \in \Ssub{B}$, $Z = \mathrm{IP}^n (X,Y)$ is $\varepsilon$-random relative to $YT$ for
    \begin{equation}
        \varepsilon = \frac{1}{2} \sqrt{2^{n - k_1 - k_2}}
    \end{equation}
    where
    \begin{equation}
        k_{1} \coloneqq -D_{2}\left( \rho_{XB},\mathds{1}_{X} \otimes \sigma_{B} \right)
        \quad \text{and} \quad 
        k_{2} \coloneqq -\log\left( \sum_{y}\tr[\left( \sigma_{B}^{1/4}\left( \cN_{T|B}^{y} \right)^{*}[\mathds{1}_{T}]\sigma_{B}^{1/4} \right)^{2}]\right).
    \end{equation}
\end{lemma}
\begin{proof}
    Let us write
    \begin{equation}
        \rho_{XB} = \sum_x \ketbra{x}_X \otimes \rho_{B \land X=x}
    \end{equation}
    and denote by $\rho_{ZYT}^\mathrm{IP} \coloneqq \mathrm{IP}^n_{ZY|XY}[\cN_{YT|B}[\rho_{XB}]]$.
    Then
    \begin{equation}
    \begin{aligned}
        & \frac{1}{2} \norm{\rho_{ZYT}^{\text{IP}} - \omega_{Z} \otimes \rho_{YT}^{\text{IP}}}_{1} \\
        =& \frac{1}{2} \sum_y \norm{\rho_{ZT \land Y=y}^{\text{IP}} - \omega_{Z} \otimes \rho_{T \land Y=y}^{\text{IP}}}_{1} \\
        =& \frac{1}{2}\sum_{y} \begin{aligned}[t]
            &\norm{\rho_{T \land Z=0,Y=y}^{\text{IP}} - \frac{1}{2}\left( \rho_{T \land Z = 0,Y = y}^{\text{IP}} + \rho_{T \land Z=1,Y=y}^{\text{IP}} \right)}_{1} \\
            &+ \norm{\rho_{T \land Z=1,Y=y}^{\text{IP}} - \frac{1}{2}\left( \rho_{T \land Z=0,Y=y}^{\text{IP}} + \rho_{T \land Z=1,Y=y}^{\text{IP}} \right)}_{1}
        \end{aligned} \\
     =& \frac{1}{2}\sum_{y}\norm{\rho_{T \land Z=0,Y=y}^{\text{IP}} - \rho_{T \land Z=1,Y=y}^{\text{IP}}}_{1} \\
     =& \frac{1}{2}\sum_{y}\norm{\sum_{z}\rho_{T \land Z=z,Y=y}^{\text{IP}} (-1)^{z}}_{1} \\
     =& \frac{1}{2}\sum_{y}\norm{\sum_{x}\cN_{T|B}^{y}[\rho_{B \land X=x}] (-1)^{x \cdot y}}_{1} \\
     =& \frac{1}{2} \sum_{y} \max_{-\mathds{1} \leq \Lambda^{y} \leq \mathds{1}} \tr[\Lambda_{T}^{y}\sum_{x}\cN_{T|B}^{y}[\rho_{B \land X=x}]( - 1)^{x \cdot y}] \\
     =& \frac{1}{2}\max_{-\mathds{1} \leq \Lambda^{y} \leq \mathds{1}}\tr[\sum_{x,y}\rho_{B \land X=x}\left( \cN_{T|B}^{y} \right)^{*}[ \Lambda_{T}^{y}](-1)^{x \cdot y}] \\
     =& \frac{1}{2}\max_{-\mathds{1} \leq \Lambda^{y} \leq \mathds{1}}\tr[ \sum_{x}\sigma_{B}^{-1/4}\rho_{B \land X=x}\sigma_{B}^{-1/4}\left( \sum_{y}\sigma_{B}^{1/4}\left( \cN_{T|B}^{y} \right)^{*}\left[ \Lambda_{T}^{y} \right]\sigma_{B}^{1/4} (-1)^{x \cdot y} \right) ]
    \end{aligned}
    \end{equation}
    holds for any $\sigma_{B}$ with $\mathrm{supp}(\rho_{B}) \subseteq \mathrm{supp}(\sigma_{B})$ (and is bounded by $+\infty$ otherwise). Let us define the Hermitian operators
    \begin{equation}
    \begin{aligned}
        P_{XB} \coloneqq &\sigma_{B}^{- 1/4}\rho_{XB}\sigma_{B}^{- 1/4}, \\
        Q_{XB} \coloneqq &\sum_{x}\ketbra{x}_{X} \otimes \left( \sum_{y}\sigma_{B}^{1/4}\left( \cN_{T|B}^{y} \right)^{*}[\Lambda_{T}^{y}]\sigma_{B}^{1/4}( - 1)^{x \cdot y} \right).
    \end{aligned}
    \end{equation}
    The Cauchy-Schwarz inequality for the Hilbert-Schmidt inner product gives 
    \begin{equation}
    \begin{aligned}
        & \abs{\tr[ \sum_{x}\sigma_{B}^{-1/4}\rho_{B \land X=x}\sigma_{B}^{-1/4}\left( \sum_{y}\sigma_{B}^{1/4}\left( \cN_{T|B}^{y} \right)^{*}[\Lambda_{T}^{y}] \sigma_{B}^{1/4}(-1)^{x \cdot y} \right)]} \\
        =& \abs{{\tr[P_{XB}Q_{XB}]}} \\
        \leq& \sqrt{\tr[P_{XB}^{2}]} \sqrt{\tr[Q_{XB}^{2}]}. \\
    \end{aligned}
    \end{equation}
    The term under the first square root equals 
    \begin{equation}
        \tr[P_{XB}^2] = \tr[\left( \sigma_{B}^{-1/4}\rho_{XB}\sigma_{B}^{-1/4} \right)^{2}] = 2^{D_{2}(\rho_{XB},\mathds{1}_{X} \otimes \sigma_{B})} = 2^{-k_1}.
    \end{equation}
    For the second square root, we compute 
    \begin{equation}
    \begin{aligned}
        \tr[Q_{XB}^2]
        =& \sum_{x}\tr[\left( \sum_{y}\sigma_{B}^{1/4}\left( \cN_{T|B}^{y} \right)^{*}[\Lambda_{T}^{y}] \sigma_{B}^{1/4}(-1)^{x \cdot y} \right)^{2}] \\
        =&\sum_{x,y,y'}\tr[ \left( \sigma_{B}^{1/4}\left( \cN_{T|B}^{y} \right)^{*} [\Lambda_{T}^{y}] \sigma_{B}^{1/4} \right) \left( \sigma_{B}^{1/4}\left( \cN_{T|B}^{y'} \right)^{*}[\Lambda_{T}^{y'}]\sigma_{B}^{1/4} \right) (-1)^{x \cdot (y + y')}] \\
        =& \sum_{y,y'}\tr[\left( \sigma_{B}^{1/4}\left( \cN_{T|B}^{y} \right)^{*}[\Lambda_{T}^{y}] \sigma_{B}^{1/4} \right)\left( \sigma_{B}^{1/4}\left( \cN_{T|B}^{y'} \right)^{*}[\Lambda_{T}^{y'}]\sigma_{B}^{1/4} \right) \sum_{x}(-1)^{x \cdot (y + y')}] \\
        =& \sum_{y,y'}\tr[\left( \sigma_{B}^{1/4}\left( \cN_{T|B}^{y} \right)^{*}[\Lambda_{T}^{y}] \sigma_{B}^{1/4} \right)\left( \sigma_{B}^{1/4}\left( \cN_{T|B}^{y'} \right)^{*}[\Lambda_{T}^{y'}] \sigma_{B}^{1/4} \right)2^{n}\delta_{y = y'}] \\
        =& 2^{n}\sum_{y}\tr[\left( \sigma_{B}^{1/4}\left( \cN_{T|B}^{y} \right)^{*}[\Lambda_{T}^{y}] \sigma_{B}^{1/4} \right)^{2}]. \\
    \end{aligned}
    \end{equation}
    Next, we decompose $\Lambda_T^{y}$ into its positive and negative parts as $\Lambda_T^{y} = \Lambda_T^{y,+} - \Lambda_T^{y,-}$. Applying \cref{lem:tr_square_ineq} gives
    \begin{equation} \label{eq:N_adj_op_ineq}
         \tr[\left( \sigma_{B}^{1/4}\left( \cN_{T|B}^{y} \right)^{*}[\Lambda_{T}^{y}] \sigma_{B}^{1/4} \right)^{2}] \leq \tr[\left( \sigma_{B}^{1/4} \left(\cN_{T|B}^y\right)^*[\Lambda_T^{y,+} + \Lambda_{T}^{y,-}] \sigma_{B}^{1/4} \right)^2].
    \end{equation}
    By the complete positivity of $\left(\cN_{T|B}^y\right)^*$, we have
    \begin{equation}
        \left(\cN_{T|B}^y\right)^*[\Lambda_T^{y,+} + \Lambda_{T}^{y,-}] \leq \left(\cN_{T|B}^y\right)^*[\mathds{1}_T],
    \end{equation}
    where we used that $\Lambda_{T}^{y,+} + \Lambda_{T}^{y,-} = \abs{\Lambda_{T}^y} \leq \mathds{1}_T$. Inserting this into \cref{eq:N_adj_op_ineq} gives
    \begin{equation}
         \tr[\left( \sigma_{B}^{1/4}\left( \cN_{T|B}^{y} \right)^{*}[\Lambda_{T}^{y}] \sigma_{B}^{1/4} \right)^{2}] \leq \tr[\left( \sigma_{B}^{1/4} \left(\cN_{T|B}^y\right)^*[\mathds{1}_T] \sigma_{B}^{1/4} \right)^2].
    \end{equation}
    Putting everything together, we find for the second square root that
    \begin{equation}
        \tr[Q_{XB}^2] \leq 2^n \sum_y \tr[\left( \sigma_{B}^{1/4} \left(\cN_{T|B}^y\right)^*[\mathds{1}_T] \sigma_{B}^{1/4} \right)^2] = 2^{n-k_2}.
    \end{equation}
    Hence, in total
    \begin{equation}
    \begin{aligned}
        \abs{\tr[P_{XB}Q_{XB}]} \leq \sqrt{2^{n - k_{1} - k_{2}}}
    \end{aligned}
    \end{equation}
    and the lemma follows.
\end{proof}

Informally, \cref{lem:ip_ext_strong} above states that if $Y$ is produced by a sufficiently random process (quantified by $k_2$), then $X$ and $Y$ can be used to extract randomness using the inner product construction.

The expression for $k_2$ in \cref{lem:ip_ext_strong} is a bit unwieldy to work with. Fortunately, we can relate it to the Rényi entropy of order two of an appropriately chosen state, as the following lemma shows.

\begin{lemma} \label{lem:channel_min_entropy}
    Let $\cN_{YT|B}$ be an instrument and $\rho_{B}$ be a quantum state with purification $\rho_{B R}$. Then
  \begin{equation} \label{eq:k2_entropy_equality}
    -\log(\sum_y \tr[\left( \rho_B^{1/4} \left(\cN_{T|B}^y\right)^*[\mathds{1}_T] \rho_B^{1/4} \right)^2 ]) \geq H_2^\downarrow(Y|R)_{\cN[\rho]},
  \end{equation}
  and equality holds if $\cN_{YT|B}$ is trace-preserving.
\end{lemma}
\begin{proof}
    By the isometric invariance of $H_{2}^{\downarrow}$, it suffices to consider the following purification of $\rho_B$ (with $R=B'$)
    \begin{equation}
        {\hat{\rho}}_{BB'} \coloneqq \rho_{B}^{1/2}\Omega_{BB'}\rho_{B}^{1/2} = \left( \rho_{B'}^{1/2} \right)^{T}\Omega_{BB'}\left( \rho_{B'}^{1/2} \right)^{T}.
    \end{equation}
    Let us introduce
    \begin{equation}
        \sigma_{YTB'} \coloneqq \cN_{YT|B}[\hat{\rho}_{BB'}] = \sum_{y} \ketbra{y}_{Y} \otimes \sigma_{TB' \land Y=y}.
    \end{equation}
    By the CPTNI property of $\cN$, we have that
    \begin{equation}
        \sigma_{B'} = \tr_{YT}\left[ \cN_{YT|B}\left[ {\hat{\rho}}_{BB'} \right] \right] \leq \tr_{B}[\hat{\rho}_{BB'}] = \rho_{B'}^{T}.
    \end{equation}
    We compute
    \begin{equation}
    \begin{aligned}
        \sigma_{B' \land Y=y} = & \tr_{T}\left[\cN_{T|B}^{y}[{\hat{\rho}}_{BB'}] \right] \\
         = & \tr_{B}\left[\left( \cN_{T|B}^{y} \right)^{*}\left[ \mathds{1}_{T} \right]{\hat{\rho}}_{BB'} \right] \\
         = & \left( \rho_{B'}^{1/2} \right)^{T}\tr_{B}\left[ \left( \cN_{T|B}^{y} \right)^{*}\left[ \mathds{1}_{T} \right]\Omega_{BB'} \right]\left( \rho_{B'}^{1/2} \right)^{T} \\
         = & \left( \rho_{B'}^{1/2} \right)^{T}\left( \left( \cN_{T|B'}^{y} \right)^{*}[\mathds{1}_{T}] \right)^{T}\left( \rho_{B'}^{1/2} \right)^{T},
    \end{aligned}
    \end{equation}
    and hence
    \begin{equation}
        \left( \cN_{T|B'}^{y} \right)^{*}[\mathds{1}_{T}] = \rho_{B'}^{-1/2}\sigma_{B' \land Y=y}^{T}\rho_{B'}^{-1/2}.
    \end{equation}
    Inserting this expression into the LHS of \cref{eq:k2_entropy_equality} gives
    \begin{align*}
         & \sum_{y}\tr[\left( \rho_{B}^{1/4}\left( \cN_{T|B}^{y} \right)^{*}[\mathds{1}_{T}] \rho_{B}^{1/4} \right)^{2}] \\
         =& \sum_{y}\tr[\left( \rho_{B}^{- 1/4}\sigma_{B \land Y=y}^{T}\rho_{B}^{- 1/4} \right)^{2}] \\
         =& \sum_{y}\tr[\left( \left( \sigma_{B \land Y=y}^{T} \right)^{1/2}\rho_{B}^{- 1/2}\left( \sigma_{B \land Y=y}^{T} \right)^{1/2} \right)^{2}] \\
         \leq& \sum_{y}\tr[\left( \left( \sigma_{B' \land Y=y}^{T} \right)^{1/2}\left( \sigma_{B'}^{T} \right)^{- 1/2}\left( \sigma_{B' \land Y=y}^{T} \right)^{1/2} \right)^{2}] \\
         =& \sum_{y}\tr[\left( \sigma_{B'}^{- 1/4}\sigma_{B' \land Y=y}\sigma_{B'}^{- 1/4} \right)^{2}] \\
         =& 2^{- H_{2}^{\downarrow}\left( Y|B' \right)_{\cN\left[ \hat{\rho} \right]}},
    \end{align*}
    where the inequality follows from $\sigma_{B}^T \leq \rho_B$ and the operator anti-monotonicity of $x \mapsto x^{- 1/2}$ (see, for instance, \cite[Table 2.2]{Tomamichel_2016}). For trace-preserving channels, we have that $\sigma_{B}^T = \rho_B$ and the inequality above becomes an equality.
\end{proof}
Combining \cref{lem:ip_ext_strong,lem:channel_min_entropy} gives us the main result of this section.

\begin{theorem} \label{thm:channel_ext}
    The function $\mathrm{IP}^n$ is a $(k_1, k_2, \varepsilon)$ two-process extractor, strong in $Y$, with
    \begin{equation}
        \varepsilon = \frac{1}{2} \sqrt{2^{n - k_1 - k_2}}.
    \end{equation}
\end{theorem}
\begin{proof}
    Let $\rho^\mathrm{out}_{XYST}$ be as in \cref{def:two_process_extractor}. Define ${\hat{\rho}}_{XSB} \coloneqq \cM_{XS|A}\left[ \rho_{AB} \right]$. Applying \cref{lem:ip_ext_strong} (with
    $\sigma_{SB} = {\hat{\rho}}_{SB}$) to ${\hat{\rho}}_{XSB}$ and
    $\mathcal{I}_{S} \otimes \cN_{YT|B}$ gives
    \begin{equation}
        \frac{1}{2} \norm{\mathrm{IP}^n_{ZY|XY}[\rho_{XYST}^{\text{out}}] - \omega_{Z} \otimes \rho_{YST}^{\text{out}}}_{1} \leq \frac{1}{2} \sqrt{2^{n - k'_{1} - k'_{2}}},
    \end{equation}
    with
    \begin{equation}
        k'_{1} = H_{2}^{\downarrow}(X|SB)_{\hat{\rho}} = H_{2}^{\downarrow}(X|SB)_{\cM[\rho]}
    \end{equation}
    and
    \begin{equation}
        k'_{2} = - \log(\sum_{y}\tr[\left( {\hat{\rho}}_{SB}^{1/4}\left( \mathcal{I}_{S} \otimes \cN_{T|B}^{y} \right)^{*}[\mathds{1}_{ST}] {\hat{\rho}}_{SB}^{1/4} \right)^{2}]).
    \end{equation}
    For $k_1'$, we immediately have
    \begin{equation}
        H_{2}^{\downarrow}(X|SB)_{\cM[\rho]} \geq H_{\min}(X|SB)_{\cM[\rho]} \geq k_1,
    \end{equation}
    where the first inequality follows from \cref{lem:H2_Hmin_bound}.
    For $k_2'$, consider the Stinespring dilation (see \cref{lem:stinespring}) $K_{SR|A}$ of $\tr_{X} \circ \cM_{XS|A}$. This means that $\sigma_{SRB} \coloneqq K_{SR|A} \rho_{AB} K_{SR|A}^*$ is a purification of $\hat{\rho}_{SB}$. Hence, by \cref{lem:channel_min_entropy}
    \begin{equation}
        - \log(\sum_{y}\tr[\left( {\hat{\rho}}_{SB}^{1/4}\left( \mathcal{I}_{S} \otimes \cN_{T|B}^{y} \right)^{*}[\mathds{1}_{ST}] {\hat{\rho}}_{SB}^{1/4} \right)^{2}]) \geq H_{2}^{\downarrow}\left( Y|R \right)_{\cN[\sigma]}.
    \end{equation}
    We can bound
    \begin{equation}
        H_{2}^{\downarrow}(Y|R)_{\cN[\sigma]} \geq H_{2}^{\downarrow}(Y|SR)_{\cN[\sigma]} \geq H_{\min}(Y|SR)_{\cN[\sigma]} \geq H_{\min}(Y|A)_{\cN[\rho]} \geq k_{2},
    \end{equation}
    where we used the data-processing inequality for $H_2^\downarrow$, \cref{lem:H2_Hmin_bound}, and that the min-entropy can only increase when applying $K_{SR|A}$.
\end{proof}
We conclude this section with two remarks regarding \cref{thm:channel_ext}.
\begin{remark}[Tightness of \cref{thm:channel_ext}] \label{rem:ip_tightness}
    The bound in \cref{thm:channel_ext} matches the classical bound shown in \cite{Chor_1988, Dodis_2004}. Furthermore, one can easily see that it is tight. For this, consider two bitstrings $X$ and $Y$ of length~$n$, such that $X$ is uniform on the first $n/2$ bits but fixed to zero on the second $n/2$ bits, whereas $Y$ is fixed to zero on the first $n/2$ bits but uniform on the second $n/2$ bits. Then clearly $X \cdot Y = 0$ and, hence, the inner-product construction fails.
\end{remark}

\begin{remark}[Relation to \cref{lem:ip_ext_strong}] \label{rem:channel_ext}
    \Cref{thm:channel_ext} and \cref{lem:ip_ext_strong} allow for randomness extraction in slightly different setups. However, as shown in \cref{sec:alternative_model}, the two setups are equivalent.
\end{remark}

\section{Extracting multiple bits}
\label{sec:deor_security}
The results from the previous section can be extended to multiple output bits using a construction proposed by Dodis et al.~\cite{Dodis_2004}. For this, define the following family of functions.

\begin{definition}[Dodis et al.'s construction \cite{Dodis_2004}]
    Let $\mathcal{K} = \{K_i\}_{i=1}^m$ be a set of $n \times n$ matrices with entries in $\left\{ 0,1 \right\}$ such that for any $0 \neq s \in \left\{ 0,1 \right\}^{m}$ it holds that
    \begin{equation} \label{eq:deor_matrix_condition}
        \rank\left( \sum_{i = 1}^m s_{i}K_{i} \right) \geq n - r
    \end{equation}
    for some $r \in \mathds{N}$. The function $\text{DEOR}^{\mathcal{K}}:\left\{ 0,1 \right\}^{n} \times \left\{ 0,1 \right\}^{n} \rightarrow \left\{ 0,1 \right\}^{m}$ is defined as
    \begin{equation} \label{eq:deor_extractor}
        \mathrm{DEOR}^{\mathcal{K}}(x,y) \coloneqq (x^{T}K_{1}y, \ldots, x^{T}K_{m}y).
    \end{equation}
    In \cref{eq:deor_matrix_condition,eq:deor_extractor}, addition is taken modulo $2$.
\end{definition}

\begin{remark}[Practicality of $\mathrm{DEOR}^\mathcal{K}$]
    As shown in \cite{Dodis_2004}, there exist collections of matrices with $r = 0$ (for any $m \leq n$). Furthermore, for $r=1$, there are efficient implementations running in time $\mathcal{O}(n \log n)$ \cite{Foreman_2025} (whenever $m \leq n$ and $n$ is a prime with $2$ as a primitive root).
\end{remark}

The idea behind the proof is to reduce the analysis of the $\mathrm{DEOR}^\mathcal{K}$ construction to the inner product construction $\mathrm{IP}^n$. The main tool for this is the classical-quantum XOR Lemma shown in \cite[Lemma 3]{Kasher_2010}.

\begin{lemma}[{Classical-quantum XOR Lemma, \cite[Lemma 3]{Kasher_2010}}]
    \label{lem:xor_lemma}
    Let $\rho_{ZE}$ be a cq state where $Z$ is a bitstring of length $m$. Then
    \begin{equation}
        \norm{\rho_{ZE} - \omega_{Z} \otimes \rho_{E}}_{1}^{2} \leq 2^{m}\sum_{s \neq 0} \norm{\rho_{(s \cdot Z)E} - \omega_{Z'} \otimes \rho_{E}}_{1}^{2},
    \end{equation}
    where the summation runs over all $0 \neq s \in \{0,1\}^m$ and $Z'$ is a one bit system.
\end{lemma}

Our proof will rely on the fact that applying a high rank matrix to a bitstring does not decrease its entropy too much. This is the content of the following lemma.
\begin{lemma}[{\cite[Proposition 2.2.3]{Miller_2025}}] \label{lem:mat_mul_ent_change}
    Let $K$ be a $n \times n$ matrix with entries in $\{0,1\}$. Let $\rho_{XR}$ be a cq state where $X$ is a bitstring of length $n$ and assume that $\rank(K) \geq n - r$. Then
    \begin{equation}
        H_{\min}((KX)|R)_{\rho} \geq H_{\min}(X|R)_{\rho} - r.
    \end{equation}
    Here, $KX$ denotes the random variable which is obtained after applying the matrix $K$ to the bitstring~$X$.
\end{lemma}

We are now ready to show the main result of this section.
\begin{theorem} \label{thm:dodis_extractor}
    $\mathrm{DEOR}^\mathcal{K}$ is a $(k_1, k_2, \varepsilon)$ two-process extractor, strong in $Y$, with
    \begin{equation}
        \varepsilon = \frac{1}{2} \sqrt{2^{2m + n + r- k_1- k_2}}.
    \end{equation}
\end{theorem}
\begin{proof}
    Let $\rho^\mathrm{out}_{XYST}$ be as in \cref{def:two_process_extractor} and let us denote $\rho^\mathrm{DEOR}_{ZYST} \coloneqq \mathrm{DEOR}^\mathcal{K}_{ZY|XY}[\rho^{\mathrm{out}}_{XYST}]$.
    Applying the XOR-Lemma (\cref{lem:xor_lemma}), we have that 
    \begin{equation} \label{eq:deor_xor_bound}
        \begin{aligned}
        \norm{\rho_{ZYST}^{\mathrm{DEOR}} - \omega_{Z} \otimes \rho_{YST}^{\mathrm{out}}}_{1}^{2}
        =& \norm{\rho_{ZYST}^{\mathrm{DEOR}} - \omega_{Z} \otimes \rho_{YST}^{\mathrm{DEOR}}}_{1}^{2} \\
        \leq & 2^{m}\sum_{s \neq 0}\norm{\rho_{(s \cdot Z)YST}^{\mathrm{DEOR}} - \omega_{Z'} \otimes \rho_{YST}^{\mathrm{DEOR}}}_{1}^{2} \\
         = & 2^{m}\sum_{s \neq 0}\norm{\mathrm{IP}^n_{Z'Y|XY}[\rho_{(K_{s}^T X)YST}^{\mathrm{out}}] - \omega_{Z'} \otimes \rho_{YST}^{\mathrm{out}}}_{1}^{2},
        \end{aligned}
    \end{equation}
    where we introduced $K_{s} = \sum_{i} s_{i}K_{i}$. We now note that by assumption $\rank(K_{s}^T) = \rank(K_{s}) \geq n - r$, and therefore by \cref{lem:mat_mul_ent_change}, $H_{\min}((K_{s}^TX)|B)_{\cM[\rho]} \geq H_{\min}(X|B)_{\cM[\rho]} - r \geq k_{1} - r$.
    Hence, we can apply \cref{thm:channel_ext} to bound 
    \begin{equation}
        \norm{\mathrm{IP}^n_{Z'Y|XY}[\rho_{(K_{s}^T X)YST}^{\mathrm{out}}] - \omega_{Z'} \otimes \rho_{YST}^{\mathrm{out}}}_{1}^{2} \leq 2^{n + r - k_{1} - k_{2}}
    \end{equation}
    for all $s \neq 0$. Inserting this into \cref{eq:deor_xor_bound} gives
    \begin{equation}
    \begin{aligned}
        \frac{1}{2}\norm{\rho_{ZYST}^{\mathrm{DEOR}} - \omega_{Z} \otimes \rho_{YST}^{\mathrm{out}}}_{1}
        \leq& \frac{1}{2}\sqrt{2^{m} \cdot 2^{m} \cdot 2^{n + r - k_{1} - k_{2}}} \\
        =& \frac{1}{2}\sqrt{2^{2m + n + r - k_{1} - k_{2}}},
    \end{aligned}
    \end{equation}
    which is the claimed bound.
\end{proof}

\begin{remark}[Tightness of \cref{thm:dodis_extractor}]
    Classically, the $\mathrm{DEOR}^\mathcal{K}$ extractor is known to be secure with $\varepsilon = \frac{1}{2} \sqrt{2^{m + n + r - k_1 - k_2}}$ \cite{Dodis_2004}. Compared to the bound in \cref{thm:dodis_extractor}, this allows for the extraction of twice as many random bits (due to the missing factor $2$ in front of $m$). The main technical reason for the difference is that the purely classical XOR Lemma does not have the $2^m$ prefactor from \cref{lem:xor_lemma}. We conjecture that one can achieve the same bound as in the classical case. Note that even for (conditionally) independent quantum states, this was shown only recently in \cite{Miller_2025}.
\end{remark}

\section{Smoothing} \label{sec:smoothing}
In practice, it can be difficult (or even impossible) to find good lower-bounds on $H_{\min}$. To avoid this issue, one often relaxes the min-entropy to its smoothed variant $H_{\min}^\varepsilon$.
The main technical hurdle is that $H_{\min}^\varepsilon(X|SB)_{\cM[\rho]} \geq k_1$ only guarantees that there exists a state of min-entropy $k_1$ which is $\varepsilon$ close to $\cM[\rho]$.
However, to \cref{def:two_process_extractor} requires a channel $\tilde{\cM}$ such that $\tilde{\cM}[\rho]$ has min-entropy $k_1$. Therefore, we wish to move the smoothing from the channel output onto the channel itself. This is done in the following lemma.

\begin{lemma} \label{lem:channel_smoothing}
    Let $\rho_{AR}$ be a pure quantum state and $\mathcal{E}_{BS|A}$ be a channel. Assume that
    $H_{\min}^{\varepsilon}\left( B|SR \right)_{\mathcal{E}[\rho]} \geq k$.
    Then, there exists a sub-normalized channel
    $\tilde{\mathcal{E}}_{BS|A}$ such that
    
    \begin{enumerate}
    \item
      $P\left( \mathcal{E}_{BS|A}\left[ \rho_{AR} \right], \tilde{\mathcal{E}}_{BS|A}\left[ \rho_{AR} \right] \right) \leq 4\varepsilon$ and
    \item
      $H_{\min}\left( B|SR \right)_{\tilde{\mathcal{E}}[\rho]} \geq k - \log \left( \frac{2}{\varepsilon^2} + \frac{1}{\tr[\rho] - \varepsilon} \right)$.
    \end{enumerate}
    Furthermore, the channel $\tilde{\mathcal{E}}$ is classical on the same systems as $\mathcal{E}$.
\end{lemma}
\begin{proof}
    The main idea is to use a weighted version of the Choi-Jamiołkowsi isomorphism \cite{Choi_1975,Jamiolkowski_1972}. More precisely, first we define a Choi state, then we use the guarantee on $H_{\min}^\varepsilon$ to find a smoothed Choi state, and finally we use the inverse isomorphism to define our smoothed channel $\tilde{\mathcal{E}}$.
    Therefore, let us define
    \begin{equation}
        \gamma_{BSA} \coloneq \mathcal{E}_{BS|A'}\left[ \rho_{A'}^{1/2}\Omega_{A'A}\rho_{A'}^{1/2} \right].
    \end{equation}
    Note that by the trace non-increasing property of $\mathcal{E}$, we have that $\gamma_{A} \leq \rho_{A}^{T}$.
    By \cref{lem:unopt_smooth_min_ent_bound}, we have that
    \begin{equation}
    \begin{aligned}
        k' \coloneqq H_{\min}^{\downarrow,2\varepsilon}(B|SR)_{\mathcal{E}[\rho]} \geq& H_{\min}^{\varepsilon}(B|SR)_{\mathcal{E}[\rho]} - \log \left( \frac{2}{\varepsilon^2} + \frac{1}{\tr[\rho] - \varepsilon} \right) \\
        \geq& k - \log \left( \frac{2}{\varepsilon^2} + \frac{1}{\tr[\rho] - \varepsilon} \right).
    \end{aligned}
    \end{equation}
    Hence, we can find a state ${\tilde{\gamma}}_{BSA}$ such that\footnote{Technically, we only assume that such a state ${\tilde{\rho}}_{BSR}$ exists for the input $\rho_{AR}$. However, we have that $\rho_{AR} = V_{R|A'}\rho_{A}^{1/2}\Omega_{AA'}\rho_{A}^{1/2}V_{R|A'}^{*}$ which means that we can pick ${\tilde{\gamma}}_{BSA} = V_{R|A}^{*}{\tilde{\rho}}_{BSR}V_{R|A}$.}
    \begin{equation} \label{eq:gamma_tilde_properties}
        P(\gamma_{BSA}, \tilde{\gamma}_{BSA}) \leq 2\varepsilon \quad \mathrm{and} \quad \tilde{\gamma}_{BSA} \leq 2^{-k'} \mathds{1}_B \otimes \tilde{\gamma}_{SA}.
    \end{equation}
    We can apply \cref{lem:change_of_marginal} to $\tilde{\gamma}_{BSA}$ and $\gamma_{A}$ to find an operator $L_A \in \mathrm{Lin}(A)$ such that the state
    \begin{equation}
        \xi_{BSA} \coloneqq L_{A} \tilde{\gamma}_{BSA} L_{A}^*
    \end{equation}
    is an extension of $\gamma_{A}$ which satisfies $P(\tilde{\gamma}_{BSA}, \xi_{BSA}) = P(\tilde{\gamma}_A, \gamma_{A}) \leq P(\tilde{\gamma}_{BSA}, \gamma_{BSA}) \leq 2\varepsilon$. Note that by the second part of \cref{eq:gamma_tilde_properties}
    \begin{equation} \label{eq:xi_op_ineq}
        \xi_{BSA} \leq 2^{-k'} \mathds{1}_B \otimes L_A \tilde{\gamma}_{SA} L_A^* = 2^{-k'} \mathds{1}_B \otimes \xi_{SA}.
    \end{equation}
    Let us define the map
    \begin{equation}
    \begin{aligned}
        \tilde{\mathcal{E}}_{BS|A}\left[ S_{A} \right] \coloneq & \tr_{A}\left[ \rho_{A}^{- 1/2}\xi_{BSA}^{T_{A}}\rho_{A}^{- 1/2}S_{A} \right]
        = \tr_{A}\left[ \left( \rho_{A}^{- 1/2} \right)^{T}\xi_{BSA}\left( \rho_{A}^{- 1/2} \right)^{T}S_{A}^{T_{A}} \right].
    \end{aligned}
    \end{equation}
    Clearly, $\tilde{\mathcal{E}}$ is completely positive. We verify that it is also trace non-increasing:
    \begin{equation}
    \begin{aligned}
        \tr_{BS}\left[ \tilde{\mathcal{E}}_{BS|A}\left[ S_{A} \right] \right] = & \tr[\left( \rho_{A}^{- 1/2} \right)^{T}\xi_{BSA}\left( \rho_{A}^{- 1/2} \right)^{T}S_{A}^{T}] \\
        = & \tr[\left( \rho_{A}^{- 1/2} \right)^{T}\xi_{A}\left( \rho_{A}^{- 1/2} \right)^{T}S_{A}^{T}] \\
        \leq & \tr[ S_{A}^{T}] \\
        = & \tr[S_{A}],
    \end{aligned}
    \end{equation}
    where the inequality follows by $\xi_A = \gamma_A \leq \rho_{A}^T$. Let us compute
    \begin{equation}
    \begin{aligned}
        \tilde{\mathcal{E}}_{BS|A}\left[ \rho_{A}^{1/2}\Omega_{AA'}\rho_{A}^{1/2} \right] = & \tr_{A}\left[ \rho_{A}^{- 1/2}\xi_{BSA}^{T_{A}}\rho_{A}^{- 1/2}\rho_{A}^{1/2}\Omega_{AA'}\rho_{A}^{1/2} \right] \\
        = & \tr_{A}\left[ \xi_{BSA}^{T_{A}}\Omega_{AA'} \right] \\
        = & \xi_{BSA'}.
    \end{aligned}
    \end{equation}
    Now note that since $\rho_{AR}$ and $\rho_{A}^{1/2}\Omega_{AA'}\rho_{A}^{1/2}$ both purify $\rho_{A}$, we can write
    \begin{equation}
        \rho_{AR} = V_{R|A'}\rho_{A}^{1/2}\Omega_{AA'}\rho_{A}^{1/2}V_{R|A'}^{*}
    \end{equation}
    for some isometry $V_{R|A'}$. Hence
    \begin{equation}
        \mathcal{E}_{BS|A}\left[ \rho_{AR} \right] = V_{R|A'}\gamma_{BSA'}V_{R|A'}^{*}\quad\text{ and }\quad\tilde{\mathcal{E}}_{BS|A}\left[ \rho_{AR} \right] = V_{R|A'}\xi_{BSA'}V_{R|A'}^{*}.
    \end{equation}
    We now verify the two properties:
    \begin{enumerate}
    \item We have that 
        \begin{equation}
        \begin{aligned}
            P\left( \mathcal{E}_{BS|A}\left[ \rho_{AR} \right], \tilde{\mathcal{E}}_{BS|A}\left[ \rho_{AR} \right] \right) = & P\left( \gamma_{BSA},\xi_{BSA} \right) \\
            \leq & P\left( \gamma_{BSA},{\tilde{\gamma}}_{BSA} \right) + P\left( {\tilde{\gamma}}_{BSA},\xi_{BSA} \right) \\
            \leq & 4\varepsilon,
        \end{aligned}
        \end{equation}
        where we used isometric invariance and the triangle inequality.
    \item We have
      \begin{equation}
        \tilde{\mathcal{E}}_{BS|A}\left[ \rho_{AR} \right] = V_{R|A'}\xi_{BSA'}V_{R|A'}^{*} \leq 2^{- k'}\mathds{1}_{B} \otimes \underbrace{\left( V_{R|A'} \xi_{SA'} V_{R|A'}^{*} \right)}_{\in \Ssub{SR}},
    \end{equation}
    where the inequality follows from \cref{eq:xi_op_ineq}. Hence $H_{\min}(B|SR)_{\tilde{\mathcal{E}}[\rho]} \geq k'$.
    \end{enumerate}
    It is well-known that the optimizer for $H_{\min}^\varepsilon(B|SR)_{\mathcal{E}[\rho]}$ is classical on the same systems as $\mathcal{E}[\rho]$ \cite[Lemma 6.13]{Tomamichel_2016}. Hence, by the definition of $\tilde{\mathcal{E}}$, it inherits this structure. This concludes the proof.
\end{proof}
The following lemma is a slight variation of \cref{lem:channel_smoothing}.
\begin{lemma} \label{lem:channel_smoothing_alt}
    Let $\rho_{AR}$ be a pure quantum state and $\mathcal{E}_{BS|A}$ be a channel. Assume that
    $H_{\min}^{\varepsilon}\left( B|R \right)_{\mathcal{E}[\rho]} \geq k$.
    Then, there exists a sub-normalized channel
    $\tilde{\mathcal{E}}_{BS|A}$ such that
    \begin{enumerate}
    \item $P\left( \mathcal{E}_{BS|A}\left[ \rho_{AR} \right], \tilde{\mathcal{E}}_{BS|A}\left[ \rho_{AR} \right] \right) \leq 4\varepsilon$ and
    \item $H_{\min}\left( B|R \right)_{\tilde{\mathcal{E}}[\rho]} \geq k - \log \left( \frac{2}{\varepsilon^2} + \frac{1}{\tr[\rho] - \varepsilon} \right)$.
    \end{enumerate}
    Furthermore, the channel $\tilde{\mathcal{E}}$ is classical on the same systems as $\mathcal{E}$.
\end{lemma}
\begin{proof}
  The proof proceeds analogously to the proof of \cref{lem:channel_smoothing}. The only difference is that we now get a state $\tilde{\gamma}_{BA}$ such that
  \begin{equation}
    P(\gamma_{B A}, \tilde{\gamma}_{BA}) \leq 2\varepsilon \quad\text{and}\quad \tilde{\gamma}_{BA} \leq 2^{-k'} \mathds{1}_B \otimes \tilde{\gamma}_A.
  \end{equation}
  By \cref{lem:purified_extension}, we can find an extension $\tilde{\gamma}_{BSA}$ of $\tilde{\gamma}_{BA}$ such that
  \begin{equation}
    P(\gamma_{BSA}, \tilde{\gamma}_{BSA}) = P(\gamma_{BA}, \tilde{\gamma}_{BA}) \leq 2\varepsilon.
    \end{equation}
    Applying the arguments from \cref{lem:channel_smoothing} to $\tilde{\gamma}_{BSA}$ yields the desired statement.
\end{proof}
We now state and show the main result of this section. We treat the strong extractor case here, but analogous statements can also be made about weak extractors.
\begin{theorem} \label{thm:smoothing}
    Let $\rho_{AB}$ be a pure quantum state and $\varepsilon_1, \varepsilon_2, k_1, k_2 \geq 0$. Define $k'_i \coloneqq k_i - \log \left( \frac{2}{\varepsilon_i^2} + \frac{1}{\tr[\rho_{AB}] - \varepsilon_i} \right)$ for $i=1,2$. Let $\mathrm{Ext}: \{0,1\}^{n_1} \times \{0,1\}^{n_2} \rightarrow \{0,1\}^m$ be a $(k'_1, k'_2, \varepsilon)$ two-process extractor, strong in $Y$. Assume that $\cM_{XS|A}, \cN_{YT|B}$ are instruments such that 
    \begin{equation*}
        H_{\min}^{\varepsilon_1}(X|SB)_{\cM[\rho]} \geq k_1 \quad \text{and} \quad H_{\min}^{\varepsilon_2}(Y|A)_{\cN[\rho]} \geq k_2
    \end{equation*}
    hold.
    Define $\rho_{XYST}^{\mathrm{out}} = \left( \cM_{XS|A} \otimes \cN_{YT|B} \right)[\rho_{AB}]$. Then $Z=\mathrm{Ext}(X,Y)$ is $\tilde{\varepsilon}$-random relative to $YST$ for
    \begin{equation}
        \tilde{\varepsilon} = 8(\varepsilon_1 + \varepsilon_2) + \varepsilon.
    \end{equation}
\end{theorem}
\begin{proof}
    Applying \cref{lem:channel_smoothing} to $\cM_{XS|A}$ and \cref{lem:channel_smoothing_alt} to $\cN_{YT|B}$ gives us instruments $\tilde{\cM}_{XS|A}$ and $\tilde{\cN}_{YT|B}$ such that
    \begin{equation} \label{eq:channel_dist}
        \norm{\left( \cM_{XS|A} - \tilde{\cM}_{XS|A} \right)\left[ \rho_{AB} \right]}_{+} \leq 4\varepsilon_{1}\quad\text{ and }\quad\norm{\left( \cN_{YT|B} - \tilde{\cN}_{YT|B} \right)\left[ \rho_{AB} \right]}_{+} \leq 4\varepsilon_{2}.
    \end{equation}
    Furthermore, we have that
    \begin{equation}
        H_{\min}\left( X|SB \right)_{\tilde{\cM}[\rho]} \geq k'_{1}\quad\text{ and }\quad H_{\min}(Y|A)_{\tilde{\cN}[\rho]} \geq k'_{2}.
    \end{equation}

    Let us denote 
    \begin{equation}
    \begin{aligned}
        \rho_{ZYST}^{\mathrm{Ext}} \coloneq& \mathrm{Ext}_{ZY|XY} \circ \left(\cM_{XS|A} \otimes \cN_{YT|B} \right)[\rho_{AB}], \\
        \tilde{\rho}_{ZYST}^{\mathrm{Ext}} \coloneq& \mathrm{Ext}_{ZY|XY} \circ \left( \tilde{\cM}_{XS|A} \otimes \tilde{\cN}_{YT|B} \right)[\rho_{AB}].
    \end{aligned}    
    \end{equation}
    Note that \cref{eq:channel_dist} implies
    \begin{equation} \label{eq:opt_state_to_approx_state}
    \begin{aligned}
    \norm{\rho_{ZYST}^{\mathrm{Ext}} - {\tilde{\rho}}_{ZYST}^{\mathrm{Ext}}}_{+} \leq & \norm{\left( \cM_{XS|A} \otimes \cN_{YT|B} - \tilde{\cM}_{XS|A} \otimes \tilde{\cN}_{YT|B} \right)\left[ \rho_{AB} \right]}_{+} \\
     \leq & \norm{\left( \left( \cM_{XS|A} - \tilde{\cM}_{XS|A} \right) \otimes \cN_{YT|B} \right)\left[ \rho_{AB} \right]}_{+} \\
     & + \norm{\left( \tilde{\cM}_{XS|A} \otimes \left( \cN_{YT|B} - \tilde{\cN}_{YT|B} \right) \right)\left[ \rho_{AB} \right]}_{+} \\
     \leq & 4(\varepsilon_{1} + \varepsilon_{2}),
    \end{aligned}
    \end{equation}
    where we used the data-processing inequality and the triangle inequality. Since $\mathrm{Ext}$ is a $(k'_1, k'_2, \varepsilon)$ two-process extractor, we have that
    \begin{equation} \label{eq:smooth_deor}
        \frac{1}{2}\norm{{\tilde{\rho}}_{ZYST}^{\mathrm{Ext}} - \omega_{Z} \otimes {\tilde{\rho}}_{YST}^{\mathrm{Ext}}}_{1}^{2} \leq \varepsilon.
    \end{equation}
    Combining the bounds then yields 
    \begin{equation}
    \begin{aligned}
        \frac{1}{2} \norm{\rho_{ZYST}^{\mathrm{Ext}} - \omega_{Z} \otimes \rho_{YST}^{\mathrm{Ext}}}_{1}
        \leq& \norm{\rho_{ZYST}^{\mathrm{Ext}} - {\tilde{\rho}}_{ZYST}^{\mathrm{Ext}}}_{+} + \norm{{\tilde{\rho}}_{ZYST}^{\mathrm{Ext}} - \omega_{Z} \otimes {\tilde{\rho}}_{YST}^{\mathrm{Ext}}}_{+} \\
        &+ \norm{\omega_{Z} \otimes \left( {\tilde{\rho}}_{YST}^{\mathrm{Ext}} - \rho_{YST}^{\mathrm{Ext}} \right)}_{+} \\
        \leq& 8\left( \varepsilon_{1} + \varepsilon_{2} \right) + \varepsilon,
    \end{aligned}
    \end{equation}
    where we used the triangle inequality, \cref{eq:opt_state_to_approx_state} twice, and \cref{eq:smooth_deor}.
\end{proof}

Applying \cref{thm:smoothing} to the $\mathrm{DEOR}^\mathcal{K}$ extractor gives the following corollary.
\begin{corollary}
    Let $\rho_{AB}$ be a pure quantum state and $\cM_{XS|A}$ and $\cN_{YT|B}$ be instruments such that
    \begin{equation}
    H_{\min}^{\varepsilon_{1}}\left( X|SB \right)_{\cM[\rho]} \geq k_{1} \quad \text{and} \quad H_{\min}^{\varepsilon_{2}}\left( Y|A \right)_{\cN[\rho]} \geq k_{2}
    \end{equation}
    hold. Define $\rho_{XYST}^{\mathrm{out}} = \left( \cM_{XS|A} \otimes \cN_{YT|B} \right)[\rho_{AB}]$.
    Then, $Z = \mathrm{DEOR}_{Z|XY}^{\mathcal{K}}(X, Y)$ is $\tilde{\varepsilon}$-random relative to $YST$ for
    \begin{equation}
        \tilde{\varepsilon} = 8(\varepsilon_1 + \varepsilon_2) + \frac{1}{2}\sqrt{2^{2m + n + r - k'_{1} - k'_{2}}},
    \end{equation}
    where $k_i' \coloneqq k_i - \log \left( \frac{2}{\varepsilon_i^2} + \frac{1}{\tr[\rho_{AB}] - \varepsilon_i} \right)$ for $i = 1,2$.
\end{corollary}

\section{Relation to prior work} \label{sec:prior_work}
In this section we discuss the relation of our results to prior work on two-source extractors. In particular, we will consider classical two-source extractors \cite{Chor_1988}, the Markov model from \cite{Friedman_2016}, and the general entangled adversary model from \cite{Chung_2014}. For simplicity, we will only consider the weak extractor case, but all statements also remain valid for strong extractors.

\subsection{Classical two-source extractors} \label{subsec:cl_extractors}
As mentioned in the introduction, there is a rich history of literature on classical two-source extractors (see \cite{Chattopadhyay_2022} for a review). We begin by reproducing the definition of classical two-source extractors.
\begin{definition}[Two-source extractor \cite{Raz_05}] \label{def:two_source_extractor}
    A function $\mathrm{Ext}: \{0,1\}^{n_1} \times \{0,1\}^{n_2} \rightarrow \{0,1\}^m$ is called a \emph{$(k_1, k_2, \varepsilon)$ two-source extractor} if for all classical states $\rho_{XY} = \rho_{X} \otimes \rho_{Y}$ with $H_{\min}(X)_\rho \geq k_1$ and $H_{\min}(Y)_\rho \geq k_2$ it holds that
    $Z = \mathrm{Ext}(X,Y)$ is $\varepsilon$-random.
\end{definition}

One can easily see that applying \cref{def:two_process_extractor} to the instruments $\cM_{X|A}[\rho_A] \coloneqq \tr[\rho_A] \rho_X$ and $\cN_{Y|B}[\rho_B] \coloneqq \tr[\rho_B] \rho_Y$, gives the condition in \cref{def:two_source_extractor}. Hence, any $(k_1, k_2, \varepsilon)$ two-process extractor is a $(k_1, k_2, \varepsilon)$ two-source extractor. More interestingly, one can use two-process extractors to extract from non-independent sources, as the following lemma shows.
\begin{lemmaBreakable} \label{lem:ext_no_indep}
    Let $p(x,y)$ be an arbitrary probability distribution and $\mathrm{Ext}$ be a $(k_1, k_2, \varepsilon)$ two-process extractor. Define the states
    \begin{equation}
    \begin{aligned}
        \eta_{XB} \coloneqq& \sum_{x} p(x) \ketbra{x}_X \otimes \ketbra{\eta_x}_B &\text{with} \quad \ket{\eta_x}_B \coloneqq& \sum_{y} \sqrt{p(y|x)} \ket{y}_B, \\
        \nu_{YA} \coloneqq& \sum_{y} p(y) \ketbra{y}_Y \otimes \ketbra{\nu_y}_A &\text{with} \quad \ket{\nu_y}_A \coloneqq& \sum_{x} \sqrt{p(x|y)} \ket{x}_A.
    \end{aligned}
    \end{equation}
    If $H_{\min}(X|B)_\eta \geq k_1$ and $H_{\min}(Y|A)_\nu \geq k_2$, then $\rho_{XY} = \sum_{x,y} p(x,y) \ketbra{x,y}_{XY}$ is such that $Z = \mathrm{Ext}(X,Y)$ is $\varepsilon$-random.
\end{lemmaBreakable}
\begin{proof}
    Consider the pure state
    \begin{equation}
        \ket{\sigma}_{AB} \coloneqq \sum_{x,y} \sqrt{p(x,y)} \ket{x,y}_{AB}
    \end{equation}
    and take $\cM,\cN$ as measurements in the computational basis. Then
    \begin{equation}
        (\cM_{X|A} \otimes \cN_{Y|B})[\sigma_{AB}] = \sum_{x,y} p(x,y) \ketbra{x,y}_{XY} = \rho_{XY}.
    \end{equation}
    We compute
    \begin{equation}
    \begin{aligned}
        \cM_{X|A}[\sigma_{AB}] =& \sum_{x,y,y'} \sqrt{p(x,y)} \sqrt{p(x,y')} \ketbra{x}_X \otimes \ketbra{y}{y'}_B \\
        =& \sum_{x} p(x) \ketbra{x}_X \otimes \sum_{y,y'} \sqrt{p(y|x)} \ketbra{y}{y'}_B \sqrt{p(y'|x)} \\
        =& \eta_{XB},
    \end{aligned}
    \end{equation}
    and a similar calculation shows
    \begin{equation}
        \cN_{Y|B}[\sigma_{AB}] = \nu_{YA}.
    \end{equation}
    Since, by assumption, $H_{\min}(X|B)_{\eta} \geq k_1$ and $H_{\min}(Y|A)_{\nu} \geq k_2$ and because $\mathrm{Ext}$ is a $(k_1, k_2, \varepsilon)$ two-process extractor, we have that
    \begin{equation}
        \frac{1}{2} \norm{\mathrm{Ext}_{Z|XY}[\rho_{XY}] - \omega_Z}_1 \leq \varepsilon,
    \end{equation}
    which is the claimed statement.
\end{proof}

\begin{remark}
    In \cref{lem:ext_no_indep}, we do not place any independence assumption on $p(x,y)$, i.e, \cref{lem:ext_no_indep} allows for randomness extraction with correlated sources. The price for this are the more stringent entropy conditions $H_{\min}(X|B)_\eta \geq k_1$ instead of $H_{\min}(X|Y)_p \geq k_1$ and $H_{\min}(Y|A)_\nu \geq k_2$ instead of $H_{\min}(Y|X)_p \geq k_2$. Note that for independent $p(x,y) = p(x)p(y)$, one recovers the conditions $H_{\min}(X) \geq k_1$ and $H_{\min}(Y) \geq k_2$ as in \cref{def:two_source_extractor}.
\end{remark}

To illustrate the entropy conditions in \cref{lem:ext_no_indep}, consider the $\mathrm{IP}^n$ construction and define the following set
\begin{equation}
    S^n \coloneqq (\mathrm{IP}^n)^{-1}\{0\} = \{ (x, y) \in \{0,1\}^{n} \times \{0,1\}^{n} : x \cdot y = 0 \}.
\end{equation}
Now, define the distribution
\begin{equation}
    p(x,y) = \begin{cases}
        \frac{1}{\abs{S^n}} & \text{if } (x, y) \in S^n \\
        0 & \text{else}
    \end{cases},
\end{equation}
that is, $p(x, y)$ is uniform on $S^n$. Clearly, $\mathrm{IP}^{n}$ produces $Z=0$ with probability $1$. Hence, $\mathrm{IP}^n$ fails for the distribution $p(x, y)$. We now show that the entropies in \cref{lem:ext_no_indep} are small (which needs to be true as otherwise there would be a contradiction to \cref{thm:channel_ext}).

For this, we consider the measurement of $\eta_{X B}$ in the Hadamard basis. Let us denote by $H$ the Hadamard transform. We compute
\begin{equation}
\begin{aligned}
    H^{\otimes n} \ket{\eta_{x}}_{B} =& \sum_{y} \sqrt{p(y|x)} H^{\otimes n} \ket{y}_{B} \\
    =& \sum_{y} \sqrt{p(y|x)} \sqrt{2^{-n}} \sum_{y'} (-1)^{y \cdot y'} \ket{y'}_{B}.
\end{aligned}
\end{equation}
For $x \neq 0$, we have that $p(y|x) = 2^{-(n-1)} \delta_{x \cdot y = 0}$ and therefore
\begin{equation}
    H^{\otimes n} \ket{\eta_{x}}_{B} = 2^{-n}\sqrt{2} \sum_{y: x \cdot y = 0} \sum_{y'} (-1)^{y \cdot y'} \ket{y'}_{B}.
\end{equation}
The probability to correctly guess $x \neq 0$ given $\ket{\eta_x}$ is
\begin{equation}
    \abs{\mel{x}{H^{\otimes n}}{\eta_{x}}}^2 = \abs{2^{-n} \sqrt{2} \sum_{y: x \cdot y = 0} (-1)^{y \cdot x}}^2
    = \abs{2^{-n}\sqrt{2} 2^{n-1}}^2
    = \frac{1}{2}.
\end{equation}
For $x=0$, we have $p(y|x=0) = 2^{-n}$ and therefore
\begin{equation}
    H^{\otimes n} \ket{\eta_{x=0}}_B = 2^{-n} \sum_{y,y'} (-1)^{y \cdot y'} \ket{y'}_B.
\end{equation}
The probability to correctly guess $x = 0$ given $\ket{\eta_{x=0}}$ is
\begin{equation}
    \abs{\mel{x=0}{H^{\otimes n}}{\eta_{x=0}}}^2 = \abs{2^{-n} \sum_{y} 1}^2 = 1.
\end{equation}
Hence, given access to $B$, one can guess $x$ with probability at least $\frac{1}{2}$ and therefore \cite{Konig_2009}
\begin{equation}
    H_{\min}(X|B)_\eta \leq 1.
\end{equation}
Since the same argument also applies to $Y$ and $\nu_{YA}$, we can conclude that \cref{lem:ext_no_indep} does not allow for the extraction of randomness from $p(x,y)$ (which we already knew since $p(x,y)$ was constructed to break $\mathrm{IP}^n$).

Note that one can apply the same reasoning to other extractors $\mathrm{Ext}$. For instance, if we know that some distribution $p(x,y)$ breaks $\mathrm{Ext}$ and the entropies in \cref{lem:ext_no_indep} are $H_{\min}(X|B)_\eta = k_1$ and $H_{\min}(Y|A)_\nu = k_2$, we can conclude that $\mathrm{Ext}$ cannot be a $(k_1, k_2, \varepsilon)$ two-process extractor (although it might still be a $(k_1, k_2, \varepsilon)$ two-source extractor). 

\subsection{Markov model} \label{subsec:markov_model}

In \cite{Friedman_2016}, the authors introduce the \emph{Markov model}. As the name suggests, the Markov model considers ccq states $\rho_{XYC}$ such that the Markov chain condition $X \leftrightarrow C \leftrightarrow Y$ is satisfied, i.e., $I(X:Y|C)_\rho = 0$. Intuitively, this condition can be understood as requiring that $X$ and $Y$ are independent when conditioned on $C$ \cite{Hayden_2004}. In \cite{Friedman_2016} they introduce the following definition.
\begin{definition}[Markov model] \label{def:markov_model}
    A function $\mathrm{Ext}: \{0,1\}^{n_1} \times \{0,1\}^{n_2} \rightarrow \{0,1\}^m$ is said to be a \emph{$(k_1, k_2, \varepsilon)$ two-source extractor in the Markov model} if, for any state $\rho_{XYC}$ satisfying the Markov chain condition $X \leftrightarrow C \leftrightarrow Y$ with $H_{\min}(X|C)_\rho \geq k_1$ and $H_{\min}(Y|C)_\rho \geq k_2$, we have
    that $Z = \mathrm{Ext}(X, Y)$ is $\varepsilon$-random relative to $C$.
\end{definition}

Next, we show how the Markov model in \cref{def:markov_model} can be seen as a special case of our model.
\begin{proposition} \label{prop:markov_is_special_case}
    Any $(k_1, k_2, \varepsilon)$ two-process extractor is also a $(k_1, k_2, \varepsilon)$ extractor in the Markov model.
\end{proposition}
\begin{proof}
    Let $\mathrm{Ext}: \{0,1\}^{n_1} \times \{0,1\}^{n_2} \rightarrow \{0,1\}^m$ be a $(k_1, k_2, \varepsilon)$ two-process extractor. Consider a state $\rho_{XYC}$ such that $X \leftrightarrow C \leftrightarrow Y$ and $H_{\min}(X|C)_\rho \geq k_1$ and $H_{\min}(Y|C)_\rho \geq k_2$.
    Such a state can be decomposed as~\cite[Theorem 6]{Hayden_2004}
    \begin{equation}
        \rho_{XYC} \cong \sum_w p(w) \rho_{X C_L}^w \otimes \rho_{Y C_R}^w \otimes \ketbra{w}_W =: \rho_{X Y C_L C_R W},
    \end{equation}
    where $\cong$ means that there exists an isometry $V_{C_LC_RW|C}$ mapping the LHS to the RHS. Define the measure and prepare channels
    \begin{equation}
        \cM_{X C_L|W}[\rho_W] \coloneqq \sum_w \rho_{X C_L}^w \mel{w}{\rho_W}{w} \quad \text{and} \quad
        \cN_{Y C_R|W}[\rho_W] \coloneqq \sum_w \rho_{Y C_R}^w \mel{w}{\rho_W}{w}
    \end{equation}
    and the pure state
    \begin{equation}
        \ket{\sigma}_{W_1 W_2 W} = \sum_w \sqrt{p(w)} \ket{w}_{W_1} \otimes \ket{w}_{W_2} \otimes \ket{w}_W.
    \end{equation}
    Then $\rho_{X Y C_L C_R W} = \left(\cM_{X C_L|W_1} \otimes \cN_{Y C_R|W_2}\right)[\sigma_{W_1 W_2 W}]$. We compute
    \begin{equation}
        H_{\min} (X|C_L W_2 W)_{\cM[\sigma]} = H_{\min} (X|C)_\rho \geq k_1
    \end{equation}
    and similarly
    \begin{equation}
        H_{\min} (Y|C_R W_1 W)_{\cN[\sigma]} = H_{\min} (Y|C)_\rho \geq k_2.
    \end{equation}
    Since $\mathrm{Ext}$ is a $(k_1, k_2, \varepsilon)$ two-process extractor, we know that the state 
    \begin{equation}
        \rho_{Z C_L C_R W}^{\mathrm{Ext}} \coloneqq \left(\mathrm{Ext}_{Z|XY} \circ \cM_{XC_L|W_1} \otimes \cN_{Y C_R|W_2} \right)[\sigma_{W_1 W_2 W}] = \mathrm{Ext}_{Z|XY}[\rho_{XY C_L C_R W}]
    \end{equation}
    satisfies
    \begin{equation}
        \frac{1}{2} \norm{\rho_{ZC_L C_R W}^{\mathrm{Ext}} - \omega_Z \otimes \rho_{C_L C_R W}^{\mathrm{Ext}}}_1 \leq \varepsilon,
    \end{equation}
    which, by the isometric invariance of the trace distance, is exactly the condition of \cref{def:markov_model} and hence $\mathrm{Ext}$ is also a $(k_1, k_2, \varepsilon)$ extractor in the Markov model.
\end{proof}

Next, we show that, for separable inputs $\rho_{AB}$, an extractor in the Markov model can be used to extract randomness from $(\cM_{XS|A} \otimes \cN_{YT|B})[\rho_{AB}]$.
\begin{lemma} \label{lem:sep_inputs}
    Let $\mathrm{Ext}$ be a $(k_1, k_2, \varepsilon)$ extractor in the Markov model, $\rho_{AB} = \sum_w p(w) \rho_A^w \otimes \rho_B^w$ be a separable state, and $\cM_{XS|A}, \cN_{YT|B}$ be instruments. Define $\rho_{ABW} \coloneqq \sum_w p(w) \rho_A^w \otimes \rho_B^w \ketbra{w}_W$ and assume that that $H_{\min}(X|SW)_{\cM[\rho]} \geq k_1$ and $H_{\min}(Y|TW)_{\cN[\rho]} \geq k_2$ hold. Then, the state $\rho^{\mathrm{out}}_{XYST} \coloneqq \left(\cM_{XS|A} \otimes \cN_{YT|B} \right)[\rho_{AB}]$ is such that $Z = \mathrm{Ext}(X, Y)$ is $\varepsilon$-random relative to $ST$.
\end{lemma}
\begin{proof}
    Define the extension
    \begin{equation} \label{eq:two_process_markov}
        \rho_{XYSTW}^\mathrm{out} \coloneqq \sum_w p(w) \cM_{XS|A}[\rho_A^w] \otimes \cN_{YT|B}[\rho_B^w] \otimes \ketbra{w}_W
    \end{equation}
    which satisfies the Markov chain conditions $XS \leftrightarrow W \leftrightarrow YT$ and $X \leftrightarrow STW \leftrightarrow Y$. By assumption, we have
    \begin{equation}
        H_{\min}(X|SW)_{\cM[\sigma]} \geq k_1 \quad \text{and} \quad H_{\min}(Y|TW) \geq k_2.
    \end{equation}
    Since $T$ is independent from $XS$ when conditioned on $W$, we have that
    \begin{equation}
        H_{\min}(X|STW)_{\rho^\mathrm{out}} = H_{\min}(X|SW)_{\cM[\rho]} \geq k_1.
    \end{equation}
    Similarly, we get that
    \begin{equation}
        H_{\min}(Y|STW)_{\rho^\mathrm{out}} = H_{\min}(Y|TW)_{\cN[\rho]} \geq k_2.
    \end{equation}
    Let us define the state
    \begin{equation}
        \rho_{ZSTW}^{\mathrm{Ext}} \coloneqq \mathrm{Ext}_{Z|XY}[\rho^\mathrm{out}_{XYSTW}].
    \end{equation}
    Since $\mathrm{Ext}$ is a $(k_1, k_2, \varepsilon)$ two-source extractor in the Markov model, we can conclude that
    \begin{equation}
        \frac{1}{2} \norm{\rho^{\mathrm{Ext}}_{ZST} - \omega_Z \otimes \rho^{\mathrm{Ext}}_{ST}}_1 \leq
        \frac{1}{2} \norm{\rho^{\mathrm{Ext}}_{ZSTW} - \omega_Z \otimes \rho^{\mathrm{Ext}}_{STW}}_1 \leq \varepsilon,
    \end{equation}
    where the first inequality follows by data-processing.
\end{proof}

\begin{remark}[Strong extractors]
    \Cref{lem:sep_inputs} treats the weak extractor case. For strong extractors, we have by the data-processing inequality
    \begin{equation}
        \frac{1}{2} \norm{\rho^\mathrm{Ext}_{ZYSTW} - \omega_Z \otimes \rho^{\mathrm{Ext}}_{YSTW}}_1 \leq \frac{1}{2} \norm{\rho^\mathrm{Ext}_{ZYSW} - \omega_Z \otimes \rho^{\mathrm{Ext}}_{YSW}}_1,
    \end{equation}
    where we used that for the Markov chain $\rho^\mathrm{out}$ in \cref{eq:two_process_markov}, $T$ can be reconstructed from $W$ and $Y$. Note that $\rho^\mathrm{out}_{XYSW}$ is still a Markov chain $X \leftrightarrow SW \leftrightarrow Y$. Hence, for strong extractors, we only need the requirement $H_{\min}(Y|W)_{\cN[\rho]} \geq k_2$.
\end{remark}

Let us summarize the results of this section so far. \Cref{prop:markov_is_special_case} shows that any two-process extractor is also a two-source extractor in the Markov model (with identical parameters). Conversely, \cref{lem:sep_inputs} states that, for separable inputs, a two-source extractor in the Markov model can be used for randomness extraction in the two-process model (although the entropy conditions are slightly different). Hence, we conclude that for classically correlated (that is separable) states, the Markov model can converted into the two-process model and vice versa. The following theorem shows that for entangled inputs, this is no longer true.

\begin{theorem} \label{thm:non_markovianity}
    There exists a pure state $\rho_{AB}$ and measurements $\cM_{X|A}$, $\cN_{Y|B}$ such that any Markov state $\sigma_{XYC}$ with $\sigma_{XY} = \left(\cM_{X|A} \otimes \cN_{Y|B}\right)[\rho_{AB}]$ satisfies $H_{\min}(X|C)_{\sigma} < H_{\min}(X|B)_{\cM[\rho]}$ or $H_{\min}(Y|C)_{\sigma} < H_{\min}(Y|A)_{\cN[\rho]}$.
\end{theorem}
Informally, the lemma states that, for entangled inputs, converting from our model to the Markov model cannot be done for free. That is, in general, at least one of the two entropies will decrease.
\begin{proof}
    The proof is based on observations made in \cref{lem:ext_no_indep}. For this, let us consider the following probability distribution $p(x,y)$ where $x$ and $y$ each are bitstrings of length $2$
    \begin{equation}
    p(x,y) \coloneqq \begin{tabular}{c|cccc}
         & $00$ & $01$ & $10$ & $11$  \\ \hline
         $00$ & $1/8$ & $1/8$ & $0$ & $0$ \\
         $01$ & $0$ & $1/8$ & $1/8$ & $0$ \\
         $10$ & $0$ & $0$ & $1/8$ & $1/8$ \\
         $11$ & $1/8$ & $0$ & $0$ & $1/8$ \\
    \end{tabular}
    \end{equation}
    Take the pure state
    \begin{equation}
        \ket{\rho}_{AB} = \sum_{x,y} \sqrt{p(x,y)} \ket{x,y}_{AB}
    \end{equation}
    and $\cM_{X|A}, \cN_{Y|B}$ as measurements in the computational basis. Then, $\sigma_{XY} \coloneqq (\cM_{X|A} \otimes \cN_{Y|B})[\rho_{AB}]$ is given by $\sigma_{XY} = \sum_{x,y} p(x,y) \ketbra{x,y}_{XY}$.

    Now, we want to show that any Markov chain extension $\sigma_{XYC}$ of $\sigma_{XY}$ must have small min-entropy for either $X$ or $Y$. From \cite[Theorem 6]{Hayden_2004}, we know that $\sigma_{XYC}$ is of the form
    \begin{equation}
        \sigma_{XYC} = \bigoplus_{w} p(w) \sigma_{X C_L^w}^w \otimes \sigma_{Y C_R^w}^w.
    \end{equation}
    Let us introduce the state
    \begin{equation}
        \sigma_{XYW} = \sum_w p(w) \sigma_{X}^w \otimes \sigma_{Y}^w \otimes \ketbra{w}_W,
    \end{equation}
    which satisfies the Markov chain property $X \leftrightarrow W \leftrightarrow Y$. Furthermore, we have that
    \begin{equation}
        H_{\min}(X|C)_\sigma \leq H_{\min}(X|W)_\sigma \quad \text{and} \quad H_{\min}(Y|C)_\sigma \leq H_{\min}(Y|W)_\sigma
    \end{equation}
    by the data-processing inequality. Since $\sigma_{X}^w$ and $\sigma_{Y}^w$ are classical, we can write
    \begin{equation}
        \sigma_{X}^w = \sum_x p(x|w) \ketbra{x}_{X} \quad \text{and} \quad \sigma_{Y}^w = \sum_y p(y|w) \ketbra{y}_Y,
    \end{equation}
    for some conditional probability distributions $p(x|w)$ and $p(y|w)$. Hence, it suffices to consider classical Markov chains $X \leftrightarrow W \leftrightarrow Y$, i.e., distributions $p(x,y,w)$ with
    \begin{equation}
        p(x,y|w) = p(x|w)p(y|w) \quad \forall w.
    \end{equation}
    Due to the form of $p(x,y)$, the following properties must hold for each $w$.
    \begin{enumerate}
        \item If $p(x|w)$ is non-deterministic, then $p(y|w)$ must be deterministic and vice versa. That is, at most one of $p(x|w)$ or $p(y|w)$ can be non-deterministic.
        \item The probability $p(x|w)$ can be non-zero for at most two $x$. Similarly, the probability $p(y|w)$ can be non-zero for at most two $y$.
    \end{enumerate}
    From the first property, we know that either $X$ or $Y$ must be deterministic with probability at least $1/2$ (over $w$). Assume, without loss of generality, that $X$ is deterministic with probability $q \geq 1/2$. From the second property, we then know that for the $w$ where $X$ is not deterministic, only two values for $x$ are possible. Hence, we can guess $X$ from $W$ with probability at least
    \begin{equation}
        P_\mathrm{guess}(X|W) \geq q + (1-q)\frac{1}{2} \geq \frac{3}{4},
    \end{equation}
    where the second inequality uses that $q \geq 1/2$. Equivalently, this can be written as
    \begin{equation}
        H_{\min}(X|W)_p \leq -\log \frac{3}{4} \approx 0.41504.
    \end{equation}

    Now, one can calculate numerically\footnote{The code is available at \url{https://gitlab.phys.ethz.ch/martisan/two-process-entropies}.}
    \begin{equation}
        H_{\min}(X|B)_{\cM[\rho]} = H_{\min}(Y|A)_{\cN[\rho]} \approx 0.45689 > H_{\min}(X|W)_\sigma \geq H_{\min}(X|C)_\sigma.
    \end{equation}
\end{proof}
Interestingly, the above example is purely classical. Hence, even when there are no quantum systems at play, our model still does not reduce to the Markov model (similar observations were already made in \cref{lem:ext_no_indep}).

\subsection{General entangled adversary model}
In \cite[Section 3]{Chung_2014}, the authors introduce the general entangled adversary model (also called the GE model). We briefly reproduce their definition here.
\begin{definition}[{General entangled (GE) adversary model \cite[Definition 3.4]{Chung_2014}}] \label{def:ge_model}
    Let $\rho_{X_1 X_2 A_1 A_2} = \rho_{X_1} \otimes \rho_{X_2} \otimes \rho_{A_1 A_2}$ where $X_1$ and $X_2$ are classical systems holding $n_1$ and $n_2$ bits respectively. Consider $X_1$ and $X_2$ controlled\footnotemark{} channels $\mathcal{L}^1_{X_1E_1|X_1A_1}$, $\mathcal{L}^2_{X_2E_2|X_2A_2}$
    and a function $\mathrm{Ext}: \{0,1\}^{n_1} \times \{0,1\}^{n_2} \rightarrow \{0,1\}^m$. Define the state 
    \begin{equation}
        \rho^\mathrm{out}_{XY E_1 E_2} \coloneqq (\mathcal{L}^1_{X_1 E_1|X_1 A_1} \otimes \mathcal{L}^2_{X_2 E_2|X_2 A_2}) [\rho_{X_1} \otimes \rho_{X_2} \otimes \rho_{A_1 A_2}].
    \end{equation}
    We call $\mathrm{Ext}$ a \emph{$(k_1, k_2, \varepsilon)$ extractor in the GE model} if $\rho^\mathrm{out}_{XY E_1 E_2}$ is such that
    $Z = \mathrm{Ext}(X, Y)$ is $\varepsilon$-random relative to $E_1 E_2$ whenever
    \begin{equation} \label{eq:ge_entropy_conditions}
        H_{\min}(X_1|E_1 A_2)_{\mathcal{L}^1[\rho]} \geq k_1 \quad \text{and} \quad H_{\min}(X_2|E_2 A_1)_{\mathcal{L}^2[\rho]} \geq k_2.
    \end{equation}
\end{definition}
\footnotetext{\label{ft:one}
    This means that $\mathcal{L}^i$ acts as $\mathcal{L}^i_{X_iE_i|X_iA_i}[\rho_{X_i A_i}] = \sum_{x} \ketbra{x}_{X_i} \otimes \mathcal{L}^{i,x}_{E_i|A_i}[\mel{x}{\rho_{X_i A_i}}{x}_{X_i}]$ for some channels $\mathcal{L}^{i,x}_{E_i|A_i}$.
}

\begin{remark}
    In \cite[Section 5.2]{Friedman_2016} it was already shown that the GE model is a special case of the Markov model whenever the extractor is strong in one of the two sources. Hence, by \cref{prop:markov_is_special_case}, we can conclude that any strong two-process extractor is also a strong extractor in the GE model. Note that all results in \cite{Chung_2014} are shown for strong extractors and it is unknown whether there are any non-strong\footnotemark{} extractors which remain secure in their model.
\end{remark}
\footnotetext{ \label{ft:two}
    Any strong extractor is of course also a weak extractor. Here we explicitly mean extractors which are only weak extractors.
}

\begin{proposition} \label{prop:ge_equivalence}
    For pure input states $\rho_{A_1 A_2}$, any $(k_1, k_2, \varepsilon)$ two-process extractor is also a $(k_1, k_2, \varepsilon)$ extractor in the GE model.
\end{proposition}
\begin{proof}
    To see the equivalence, define the channels
    \begin{equation}
        \cM_{X_1 E_1|A_1}[\rho_{A_1}] \coloneqq \mathcal{L}^1_{X_1 E_1|X_1 A_1}[\rho_{X_1} \otimes \rho_{A_1}]
    \end{equation}
    and
    \begin{equation}
        \cN_{X_2 E_2|A_2}[\rho_{A_2}] \coloneqq \mathcal{L}^2_{X_2 E_2|X_2 A_2}[\rho_{X_2} \otimes \rho_{A_2}].
    \end{equation}
    That is, $\cM$ and $\cN$ prepare independent random variables $X_1$ and $X_2$ and then perform the leaking operations $\mathcal{L}^1$ and $\mathcal{L}^2$ respectively.
    The entropy conditions in \cref{eq:ge_entropy_conditions} then correspond to exactly the ones in \cref{def:two_process_extractor}. 
\end{proof}

Note, however, that our model is more general since \cref{def:ge_model} requires $\rho_{X_1 X_2} = \rho_{X_1} \otimes \rho_{X_2}$ (even after applying the leakage operations $\mathcal{L}^i$), which is not necessarily true in our model.

\begin{remark}
    In \cite{Chung_2014}, the state $\rho_{A_1 A_2}$ is assumed to be prepared by an adversary. Hence taking $\rho_{A_1 A_2}$ to be pure in \cref{prop:ge_equivalence} is not a strong restriction. 
\end{remark}

\section{Application: Device-independent randomness amplification with quantum sources} \label{sec:quantum_DIRA}
In device independent randomness amplification (DIRA), the goal is to produce (almost) uniform randomness using only a single source of imperfect randomness and two or more non-signalling devices \cite{Colbeck_2012, Kessler_2020}. The main observation behind DIRA is that there are Bell inequalities that allow for the certification of non-locality even without assuming uniform input randomness \cite{Colbeck_2012, Putz_2014}. The idea then is to use the imperfect source of randomness as the input to such a Bell test and use the observation of a Bell violation to certify the randomness of the measurement results.

In order to amplify an imperfect source of randomness, one naturally requires some measure for the quality of the input randomness. One such measure, which is frequently encountered in the literature, are probability bounded sources, also called SV sources \cite{Santha_1984}. A SV source with bias $\mu$ is defined as a sequence of random bits $X_1 \ldots X_n$ such that
\begin{equation}
    \quad \quad \quad \quad \frac{1}{2} - \mu \leq P(x_i|x^{i-1}\lambda) \leq \frac{1}{2} + \mu \quad \quad \forall i,x_i,x^{i-1},\lambda,
\end{equation}
where $\lambda$ denotes any classical information the adversary may have about the source and $x^i = x_1 \ldots x_i$. Here, we show how this can be generalized to the setting where the adversary's side information about the source may be quantum. For this, we first introduce the notion of a quantum SV source, which generalizes the classical SV source given above.

\begin{definition}[Quantum SV source] \label{def:quantum_SV_source}
    A \emph{quantum SV source with bias $\mu$} is a sequence of instruments $\{\mathcal{S}^{i}_{X_i R_i|R_{i-1}}\}_i$, where $X_i$ is a single bit, such that
    \begin{equation}
        \quad \quad \quad \quad H_{\min}(X_i|E)_{\mathcal{S}^{i}[\rho]} \geq -\log \left( \frac{1}{2} + \mu \right) \quad \quad \forall i, \rho_{R_{i - 1}E}.
    \end{equation}
\end{definition}

\begin{remark}[Relation to classical SV source]
    \Cref{def:quantum_SV_source} generalizes the classical notion of a SV-source. To see this, choose $R_i = X^i$ and
    \begin{equation}
        \mathcal{S}^{i}_{X_i X^i|X^{i-1}}[\rho_{X^{i-1}}] \coloneqq \sum_{x_i, x^{i-1}} \ketbra{x_i}_{X_i} \otimes \ketbra{x^i}_{X^i} P(x_i|x^{i-1}) \mel{x^{i-1}}{\rho_{X^{i-1}}}{x^{i-1}},
    \end{equation}
    that is, $\mathcal{S}^{i}$ receives a copy of the previous bits $X^{i-1}$, produces the next bit $X_i$ according to $P_{X_i|X^{i-1}}$, and passes along a copy of the bits $X^i$.
\end{remark}

\begin{remark}[Characterization using non-optimized min-entropy] \label{rem:non_opt_SV_condition}
  By \cite[Proposition 19]{Gour_2021}, \cref{def:quantum_SV_source} is equivalent to
  \begin{equation}
    H_\infty^\downarrow (X_i|E)_{\mathcal{S}^{i}[\rho]} \geq -\log \left(\frac{1}{2} + \mu\right) \quad\quad\forall i, \rho_{R_{i-1}E}.
  \end{equation}
\end{remark}

\begin{lemma}[Chaining of entropy] \label{lem:SV_entropy_chaining}
    Let $\{\mathcal{S}^i_{X_iR_i|R_{i-1}}\}_{i=1}^n$ be a quantum SV source with bias $\mu$. Then, for any state $\rho_{R_0 E}$, the state $\rho^\mathrm{out}_{X^n R_n E} \coloneqq (\mathcal{S}^n \circ \ldots \circ \mathcal{S}^1)[\rho_{R_0 E}]$ satisfies
    \begin{equation}
        H_{\min}(X^n|E)_{\rho^\mathrm{out}} \geq -n \log\left(\frac{1}{2} + \mu\right).
    \end{equation}
\end{lemma}
\begin{proof}
    One can directly bound
    \begin{equation}
    \begin{aligned}
        H_{\min}(X^n|E)_{\rho^\mathrm{out}}
        \geq \sum_i H_{\infty}^\downarrow(X_i|X^{i-1}E)_{\rho^{\mathrm{out}}},
    \end{aligned}
    \end{equation}
    where we used the chain rule from \cite[Proposition 5.12]{Tomamichel_2016} $n$ times. Next, using \cref{rem:non_opt_SV_condition}, we know that
    \begin{equation}
        H_{\infty}^\downarrow(X_i|X^{i-1}E)_{\rho^{\mathrm{out}}} \geq -\log\left(\frac{1}{2} + \mu\right)
    \end{equation}
    holds for all $i$. This concludes the proof.
\end{proof}

Having introduced the notion of a quantum SV source, we are now ready to illustrate how our results can be used to show the security of DIRA when the adversary holds quantum information about the source. Giving a complete security proof for a DIRA protocol is beyond the scope of this work. Instead, we will introduce the main components of DIRA security proofs and sketch how our results can be applied to prove the security of DIRA using a quantum SV source.

We will consider the following setup. Alice and Bob each use a source of imperfect randomness to choose the measurement settings in a Bell test. We model this potentially correlated sequence of measurement choices as a single SV source.\footnote{Given the spacelike separation between Alice and Bob, the order in which the measurement settings are produced is arbitrary. Nevertheless, we can, somewhat conservatively, model the whole process as two uses of a single SV source.} The measurement results of the Bell test are denoted as $X^n$. Finally, we combine $X^n$ with another $n$ pairs of bits (denoted as $Y^n$) taken from the same SV source to produce the bitstring $Z^m$.\footnote{In \cite{Kessler_2020}, the classicality of Eve's side information about the SV source is used to argue that one has a Markov chain $X^n \leftrightarrow \tilde{E} \leftrightarrow Y^n$, where $\tilde{E}$ represents all side information available to Eve.} The setup is sketched in \cref{fig:DIRA_setup}.

\begin{figure}[ht]
    \centering
    \begin{tikzpicture}[
        trace/.pic={
            \draw [thick](-0.4, 0) --   (+0.4, 0);
            \draw [thick](-0.3, 0.1) -- (+0.3, 0.1);
            \draw [thick](-0.2, 0.2) -- (+0.2, 0.2);
            \draw [thick](-0.1, 0.3)--  (0.1, 0.3);
        },
    ]
        \node[blue_box] (S1) at (2, 0) {$\mathcal{S}_1$};
        \node[blue_box] (Sn) at (5.5, 0) {$\mathcal{S}_n$};
        \node[blue_box] (Sn1) at (8.5, 0) {$\mathcal{S}_{n+1}$};
        \node[blue_box] (S2n) at (12, 0) {$\mathcal{S}_{2n}$};
        
        \node at (3.75, 0) {$\cdots$};
        \draw[thick,->,>=stealth] ([xshift=-0.7cm] S1.west) -- (S1.west);
        \draw[thick,->,>=stealth] (S1.east) -- ([xshift=0.7cm] S1.east);
        \draw[thick,->,>=stealth] ([xshift=-0.7cm] Sn.west) -- (Sn.west);
        \draw[thick,->,>=stealth] (Sn.east) -- node[pos=0.4,above]{$R_n$} (Sn1.west);
        \node at (10.25, 0) {$\cdots$};
        \draw[thick,->,>=stealth] (Sn1.east) -- ([xshift=0.7cm] Sn1.east);
        \draw[thick,->,>=stealth] ([xshift=-0.7cm] S2n.west) -- (S2n.west);

        \draw[thick,dashed] ([xshift=1.2cm,yshift=-1.3cm] Sn.east) node[below]{$\rho_{X^n R_n E}$} -- ([xshift=1.2cm,yshift=+4.3cm] Sn.east);

        \node[green_box] (E1) at (2, 1.7) {$\mathcal{E}_1$};
        \node[green_box] (En) at (5.5, 1.7) {$\mathcal{E}_n$};

        \begin{scope}[on background layer={}]
            \draw[very thick,draw=black!30!white,fill=black!5!white,rounded corners=0.1cm]
                ([xshift=-0.25cm,yshift=+0.25cm] E1.north west) --
                ([xshift=+0.25cm,yshift=+0.25cm] E1.north east) --
                ([xshift=+0.25cm,yshift=-1.2cm] S1.south east) -- node[above,text=black!85!white]{$\cM_1$}
                ([xshift=-0.25cm,yshift=-1.2cm] S1.south west) --
                cycle;

            \draw[very thick,draw=black!30!white,fill=black!5!white,rounded corners=0.1cm]
                ([xshift=-0.25cm,yshift=+0.25cm] En.north west) --
                ([xshift=+0.25cm,yshift=+0.25cm] En.north east) --
                ([xshift=+0.25cm,yshift=-1.2cm] Sn.south east) -- node[above,text=black!85!white]{$\cM_n$}
                ([xshift=-0.25cm,yshift=-1.2cm] Sn.south west) --
                cycle;
        \end{scope}

        \draw[thick,->,>=stealth] (S1.north) -- (E1.south);
        \draw[thick,->,>=stealth] (Sn.north) -- (En.south);

        \node  (dots1)at (3.75, 1.7) {$\cdots$};
        \draw[thick,->,>=stealth] ([xshift=-0.7cm] E1.west) -- (E1.west);
        \draw[thick,->,>=stealth] (E1.east) -- ([xshift=0.7cm] E1.east);
        \draw[thick,->,>=stealth] ([xshift=-0.7cm] En.west) -- (En.west);
        \draw[thick,->,>=stealth] (En.east) -- ([xshift=+0.7cm] En.east);
        \draw[thick] ([xshift=+0.7cm] En.east) pic[rotate=-90,scale=0.7] {trace};

        \draw[thick,->,>=stealth] (E1.north) -- ([yshift=0.7cm] E1.north);
        \draw[thick,->,>=stealth] (En.north) -- ([yshift=0.7cm] En.north);
        \draw[very thick,decorate,decoration={calligraphic brace,raise=0.1cm,amplitude=0.15cm}] ([xshift=-0.1cm,yshift=0.7cm] E1.north) -- node[above=0.2cm]{$X^n$} ([xshift=+0.1cm,yshift=0.7cm] En.north);

        \draw[thick,->,>=stealth] (Sn1.north) -- ([yshift=0.7cm] Sn1.north);
        \draw[thick,->,>=stealth] (S2n.north) -- ([yshift=0.7cm] S2n.north);
        \draw[very thick,decorate,decoration={calligraphic brace,raise=0.1cm,amplitude=0.15cm}] ([xshift=-0.1cm,yshift=0.7cm] Sn1.north) -- node[above=0.2cm]{$Y^n$} ([xshift=+0.1cm,yshift=0.7cm] S2n.north);

        \node[blue_box,minimum width=1.3cm] (Ext) at (10.25, 4) {$\mathrm{Ext}$};
        \draw[thick,->,>=stealth] (10.25, 2.0) -- (Ext.south);
        \draw[thick,->,>=stealth,rounded corners=0.1cm] (3.75, 3.7) -- (dots1|-Ext.west) -- (Ext.west);

        \node (Z) at ([xshift=2.8cm] Ext.east) {$Z^m$};
        \draw[thick,->,>=stealth] (Ext.east) -- (Z.west);
        \node (E) at ([yshift=-1cm] S1.west-|Z) {$E$};
        \node (Rout) at ([xshift=0.2cm] S2n.east-|Z) {$R^{\mathrm{out}}$};
        \draw[thick,->,>=stealth] (S2n.east) -- (Rout.west);
        \draw[thick,->,>=stealth] ([xshift=-0.8cm,yshift=-1cm] S1.west) -- (E.west);
        \draw[very thick,decorate,decoration={calligraphic brace,raise=0.3cm,amplitude=0.15cm}] ([yshift=+0.2cm] Z.east) -- node[right=0.5cm]{$\rho^\mathrm{out}$} ([yshift=-0.2cm] Z.east|-E.east);

        \draw[very thick,decorate,decoration={calligraphic brace,raise=0.1cm,amplitude=0.15cm}] ([xshift=-0.8cm,yshift=-1.1cm] S1.west) -- node[left=0.25cm]{$\rho^{\mathrm{in}}$} ([xshift=-0.8cm,yshift=+0.1cm] E1.west);
    \end{tikzpicture}
    \caption{
        \textbf{Diagram of a DIRA setup with a quantum source.} We model the quantum SV source as a sequence of channels $\mathcal{S}_1, \ldots, \mathcal{S}_{2n}$, producing classical random variables. The first $n$ pairs of bits are used as the input to a Bell test (green boxes) which produces the measurement results $X^n$. An additional $n$ pairs of bits $Y^n$ are produced using the same quantum SV source which, together with $X^n$, are used to extract the final random bitstring $Z^m$.}
    \label{fig:DIRA_setup}
\end{figure}
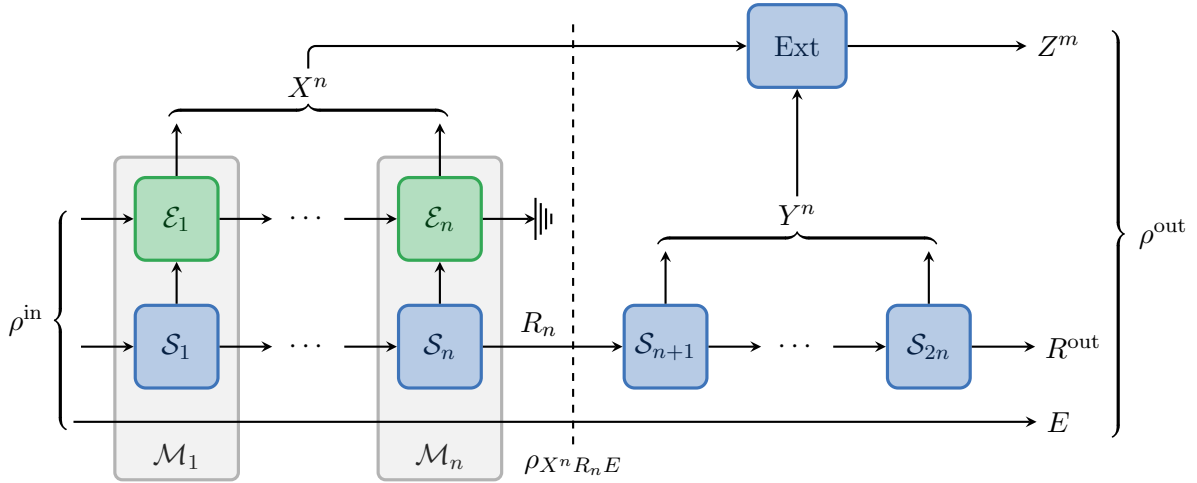

The state of the art technique for analysing DIRA protocols is based on the entropy accumulation theorem (EAT) \cite{Dupuis_2020, Dupuis_2019, Metger_2022, Kessler_2020}. Informally, the EAT is a tool which allows to bound the entropy of a quantum state which was generated by applying a sequence of channels to some initial state. The EAT then states that the overall entropy is approximately equal to the sum of the (von Neumann) entropies produced by each channel. In other words, the entropies accumulate.

It was shown in \cite[Lemma 27]{Kessler_2020} that, conditioned on observing a the violation of an appropriately chosen Bell inequality, one can bound the entropy of each channel $\cM_i$ in \cref{fig:DIRA_setup} by
\begin{equation}
    H(X_i|R_iE)_{\cM_i[\rho]} \geq h
\end{equation}
for some constant $h > 0$ which depends on the magnitude of the observed Bell violation. Let $\rho_{X^{n}R_{n}E}$ be the state after the channels $\cM_1 \ldots \cM_n$ were applied (see \cref{fig:DIRA_setup}). Then, using the generalized EAT \cite{Metger_2022} for the channels $\cM_1 \ldots \cM_n$, we have that \cite{Kessler_2020}\footnote{In \cite{Kessler_2020} the original EAT \cite{Dupuis_2020} was used to bound $H_{\min}^{\varepsilon_s}$. Here we need the generalized EAT (GEAT) \cite{Metger_2022} since we include the memory of the quantum SV source (i.e., $R_n$ in \cref{fig:DIRA_setup}) in Eve's side information, which is updated in every round. The non-signalling condition from the GEAT is clearly satisfied since in \cref{fig:DIRA_setup} there are no wires going from the green boxes to Eve. This corresponds to the assumption in \cite{Kessler_2020} that the SV source and the devices used in the Bell test are isolated.}
\begin{equation}
    H_{\min}^{\varepsilon_s}(X^{n}|R_{n}E) \geq nh - \mathcal{O}(\sqrt{n}).
\end{equation}

Let us denote $\cN = \mathcal{S}_{2n} \circ \ldots \circ \mathcal{S}_{n + 1}$. We know from \cref{lem:SV_entropy_chaining} that 
\begin{equation}
    H_{\min}(Y^n|\tilde{E})_{\cN[\sigma]} \geq -2n\log\left(\frac{1}{2}+\mu\right)
\end{equation}
holds for any $\sigma_{R_n \tilde{E}}$ and in particular for any purification of $\rho_{X^nR^nE}$ (we have a factor of $2n$ above since each $\mathcal{S}_i$ produces a pair of bits).
Hence, we can apply \cref{lem:alt_model_equivalence,thm:dodis_extractor} to obtain
\begin{equation} \label{eq:dira_security_bound}
    \frac{1}{2}\norm{\rho_{Z^{m}Y^nR^{\mathrm{out}}E}^{\text{out}} - \omega_{Z^{m}} \otimes \rho_{Y^{n}R^{\mathrm{out}}E}^{\mathrm{out}}}_{1} \leq \varepsilon_s + \frac{1}{2}\sqrt{2^{2m + 2n - nh + \mathcal{O}(\sqrt{n}) - 2nk_{2}}},
\end{equation}
where $k_{2} = - \log\left( \frac{1}{2} + \mu \right)$. This means that, for some target security parameter $\varepsilon$, one can extract
\begin{equation}
    m = \frac{1}{2} n (2k_2 + h - 2) - \log \frac{1}{2(\varepsilon - \varepsilon_s)} - \mathcal{O}(\sqrt{n})
\end{equation}
bits of uniform randomness. For this to be positive, we require that $2k_2 + h > 2$. Given that for increasing bias $\mu$, both $k_2$ and $h$ decrease, there is a maximum bias which can be tolerated.

\begin{remark}[Privatization]
    In the setup above, since the extractor in \cref{thm:dodis_extractor} is strong, one can include a copy of the output of the sources into the system $R^\mathrm{out}$. Hence, \cref{eq:dira_security_bound} then states that $Z^m$ is random even when Eve learns the output of the sources. This is also referred to as privatization \cite{Kessler_2020, Foreman_2023}.
\end{remark}

\section{Conclusions and outlook} \label{sec:conclusion}
It is essential to understand the minimal assumptions under which one can produce uniform randomness. Towards achieving this goal, researchers have shown that one can extract perfect randomness from two (conditionally) independent sources of randomness. Justifying this independence, however, is difficult as correlations are ubiquitous and there is no physical principle that prevents physical degrees of freedom at different locations from being correlated. Here, to overcome this issue, we introduce a different approach where the independence assumption is not placed on the quantum state itself but rather on the process by which it was generated. Crucially, in contrast to the independence of states, the independence of processes can be justified by physical principles such as non-signalling. We have then shown that two independent processes are sufficient for generating randomness as long as each process produces a sufficient amount of entropy.

To illustrate the versatility of this approach, we considered the example of device-independent randomness amplification (DIRA). A widely used model for the source of low-quality randomness in DIRA are SV sources \cite{Santha_1984, Colbeck_2012, Kessler_2020}. However, due to their origin in classical information theory, SV sources do not allow for quantum side information. To overcome this limitation, we generalize SV sources to a sequence of quantum channels producing only weakly biased bits. Apart from more closely matching how such sources of randomness are physically realized, this definition very naturally allows for quantum side information.

We conclude with some important open questions.
\begin{enumerate}
    \item The extractors in \cref{thm:channel_ext,thm:dodis_extractor} only work for sources such that $k_1 + k_2 > n$. This is a fairly strong (although not necessarily unrealistic) requirement. Even though our bound for the $\mathrm{IP}$ extractor is tight (see \cref{rem:ip_tightness}), better extractors are known in the classical setting (see, e.g., \cite{Chattopadhyay_2022} for an overview).
    The best known extractors for (conditionally) independent states only require sources with (poly) logarithmic min entropy \cite{Chor_1988,Chattopadhyay_2016,Friedman_2016}. It is therefore a natural question to ask what the minimal entropy requirements are to generate randomness in our model. Achieving sublinear entropy requirements would also allow for DIRA with arbitrary bias $\mu < \frac{1}{2}$ \cite{Kessler_2020}.

    \item In \cite{Friedman_2016}, it was shown that any extractor against classical side information remains secure against quantum side information in the Markov model with an exponential penalty term to the error (similar results were shown previously for seeded extractors, i.e., uniform $Y$, in \cite{Berta_2015,Berta_2017}).
    We don't know whether the same is true for our model. However, note that here the challenge seems to be far greater than in \cite{Friedman_2016} since even when there is only classical information (i.e., $S$ and $T$ are trivial), our model does not reduce to the one studied in the literature on classical two-source extractors (see \cref{subsec:cl_extractors}).

    \item Is it possible to generalize our model even further? One possible direction could be to study approximately independent channels (for some suitable approximation). This would be particularly interesting for scenarios where one does not have spacelike separation but only (imperfectly) isolated laboratories. Another direction could be to study more general scenarios where a more complicated structure is imposed on the generating process. For instance, it would be interesting to know if one can still extract randomness when each pair of bits is produced independently but some limited communication is allowed between subsequent pairs. 
    
    \item We may ask whether there is an information-theoretic criterion that can be used to decide whether a given situation fits into our model with independent channels. Note that such a characterization exists for the Markov model, namely the conditional mutual information. If it equals zero then the Markov chain condition holds \cite{Hayden_2004} (however, this criterion is not robust \cite{Ibinson_2007, Fawzi_2015}).

    \item In general, one may wish to extract randomness from more than two sources. Of particular interest in this setting is the scenario when some of the sources are faulty, i.e., they have zero min-entropy. In the classical setting, some constructions for this setup have been given in \cite{Chattopadhyay_2020}. Showing that similar results are possible in the quantum setting could enable the construction of distributed randomness beacons.
\end{enumerate}

\section*{Acknowledgments}
We thank Joseph Renes, Lukas Schmitt, and Marco Tomamichel for useful discussions.
This work is part of the project \emph{Certified Randomness Generation} (No.~20QU-1\_225171) of the Swiss State Secretariat for Education, Research and Innovation (SERI). We also acknowledge support from the Swiss National Science Foundation via project No.~20CH21\_218782, the National Centre of Competence in Research \emph{SwissMAP}, and the ETH Zurich Quantum Center. 

\bibliographystyle{halpha}
\bibliography{sources_nourl.bib}

\appendix
\section{Technical Lemmas}
\begin{lemma}[{\cite[Lemma B.3]{Dupuis_2014}}] \label{lem:change_of_marginal}
    Let $\rho_{AB} \in \Ssub{AB}$ and $\sigma_{A} \in \Ssub{A}$. Then, there exists $T_A \in \mathrm{Lin}(A)$ such that
    \begin{equation}
        \sigma_{AB} \coloneqq T_A \rho_{AB} T_A^*
    \end{equation}
    is an extension of $\sigma_A$ with $P(\rho_{AB}, \sigma_{AB}) = P(\rho_{A}, \sigma_{A})$.
\end{lemma}

\begin{lemma}[{\cite[Lemma 18]{Tomamichel_2011}}] \label{lem:unopt_smooth_min_ent_bound}
    Let $\rho_{AB} \in \Ssub{AB}$ and $0 < \varepsilon \leq \tr[\rho]$. It holds that
    \begin{equation}
        H_{\min}^{\downarrow,2\varepsilon}(A|B)_\rho \geq H_{\min}^{\varepsilon}(A|B)_\rho - \log \left( \frac{2}{\epsilon^2} + \frac{1}{\tr[\rho] - \varepsilon} \right).
    \end{equation}
\end{lemma}
The following inequality was shown in \cite[Corollary 5.10]{Tomamichel_2016} for normalized states. For completeness, we show the statement here for sub-normalized states.
\begin{lemma} \label{lem:H2_Hmin_bound}
    Let $\rho_{AB} \in \Ssub{AB}$. Then
    \begin{equation}
        H_2^\downarrow(A|B)_\rho \geq H_{\min}(A|B)_\rho.
    \end{equation}
\end{lemma}
\begin{proof}
    Let us denote $k \coloneqq H_{\min}(A|B)_\rho$. By the definition of $H_{\min}$, we know that there exists $\sigma_B \in \Ssub{B}$ such that
    \begin{equation}
        \rho_{AB} \leq 2^{-k} \mathds{1}_A \otimes \sigma_B.
    \end{equation}
    Hence,
    \begin{equation}
    \begin{aligned}
        \tr[\left(\rho_{B}^{-1/4} \rho_{AB} \rho_{B}^{-1/4}\right)^2]  =& \tr[\left(\rho_{B}^{-1/4} \rho_{AB} \rho_{B}^{-1/4}\right)\left(\rho_{B}^{-1/4} \rho_{AB} \rho_{B}^{-1/4}\right)] \\
        \leq& 2^{-k} \tr[\left(\rho_{B}^{-1/4} \sigma_B \rho_{B}^{-1/4}\right)\left(\rho_{B}^{-1/4} \rho_{AB} \rho_{B}^{-1/4}\right)] \\
        =& 2^{-k} \tr[\sigma_{B} \rho_{B}^{-1/2} \rho_{B} \rho_{B}^{-1/2}] \\
        \leq& 2^{-k} \tr[\sigma_{B}] \\
        \leq& 2^{-k}
    \end{aligned}
    \end{equation}
    and therefore,
    \begin{equation}
        H_2^\downarrow(A|B)_\rho \geq k = H_{\min}(A|B)_\rho
    \end{equation}
    as claimed.
\end{proof}

\begin{lemma} \label{lem:tr_square_ineq}
    Let $S_A \in \mathrm{Herm}(A)$ be a Hermitian operator and $S_A^{\pm}$ be positive operators such that $S_{A} = S_{A}^+ - S_{A}^-$. Then
    \begin{equation}
        \tr\bigl[S_A^{2}\bigr] \leq \tr\bigl[(S_A^+ + S_A^-)^2\bigr].
    \end{equation}
    In particular, for any $K_{B|A} \in \mathrm{Lin}(A, B)$,
    \begin{equation}
        \tr[\bigl(K_{B|A} S_{A} K_{B|A}^*\bigr)^2] \leq \tr[\bigl(K_{B|A} (S_A^+ + S_A^-) K_{B|A}^*\bigr)^2].
    \end{equation}
\end{lemma}
\begin{proof}
    We have
    \begin{equation}
    \begin{aligned}
        \tr\bigl[S_A^{2}\bigr] =& \tr[(S_A^{+} - S_A^{-})^2] \\
        =& \tr[(S_A^+)^2] - 2\underbrace{\tr[S_A^+ S_A^-]}_{\geq 0} + \tr[(S_A^-)^2] \\
        \leq& \tr[(S_A^+)^2] + 2\tr[S_A^+ S_A^-] + \tr[(S_A^-)^2] \\
        =& \tr\bigl[(S_A^+ + S_A^-)^2\bigr] \\
    \end{aligned}
    \end{equation}
    For the second statement, we can apply the above inequality with the decomposition 
    \begin{equation}
        K_{B|A} S_A K_{B|A}^* = \underbrace{K_{B|A} S_A^+ K_{B|A}^*}_{\geq 0} - \underbrace{K_{B|A} S_A^- K_{B|A}^*}_{\geq 0}
    \end{equation}
    which gives
    \begin{equation}
        \tr[\bigl(K_{B|A} S_{A} K_{B|A}^*\bigr)^2] \leq \tr[\bigl(K_{B|A} (S_A^+ + S_A^-) K_{B|A}^*\bigr)^2],
    \end{equation}
    as desired.
\end{proof}

\section{Alternative model} \label{sec:alternative_model}
The goal of this section is to show that two-process extractors can be used to extract randomness in a slightly different setup. Specifically, we consider a cq state $\rho_{XB}$ and an instrument $\cN_{YT|B}$, see \cref{fig:alt_setup}. 

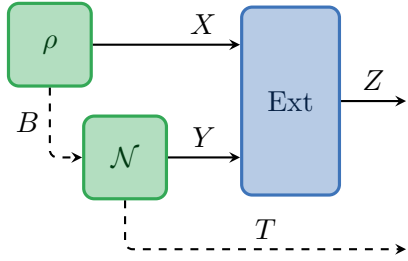
\begin{figure}[ht!]
    \centering
    \begin{tikzpicture}
        \node[green_box] (rho) at (0, 0) {$\rho$};
        \node[green_box] (N) at (+1, -1.5) {$\cN$};

        \draw[thick,->,>=stealth,dashed,rounded corners=0.1cm] (rho.south) -- node[left]{$B$} (rho.south|-N.west) -- (N.west);

        \node[blue_box, minimum height=2.5cm, minimum width=1.3cm] (Ext) at (3.2, -0.75) {$\mathrm{Ext}$};
        \draw[thick,->,>=stealth] (N.east) -- node[above]{$Y$} (N.east-|Ext.west);
        \draw[thick,->,>=stealth] (rho.east) -- node[above,pos=0.75]{$X$} (rho.east-|Ext.west);
        
        \node (Z) at ([xshift=1.0cm] Ext.east) {};
        \node (T) at ([yshift=-0.7cm] Ext.south-|Z.center) {};

        \draw[thick,->,>=stealth] (Ext.east) -- node[above]{$Z$} (Z.west);
        \draw[thick,->,>=stealth,dashed,rounded corners=0.1cm] (N.south) -- (N.south|-T.west) -- node[above]{$T$} (T.west);
    \end{tikzpicture}
    \caption{Diagramm of the alternative model studied in \cref{lem:alt_model_equivalence}. An instrument $\cN_{YT|B}$ is applied to part of a cq state $\rho_{XB}$. The function $\mathrm{Ext}$ is applied to $\rho_{XYT}^{\mathrm{out}} \coloneqq \cN_{YT|B}[\rho_{XB}]$ to extract the random bitstring $Z$.}
    \label{fig:alt_setup}
\end{figure}

The following lemma shows that two-process extractors can extract randomness from $\cN_{YT|B}[\rho_{XB}]$.
\begin{lemma} \label{lem:alt_model_equivalence}
    Let $\rho_{XB}$ be a cq state and $\cN_{YT|B}$ be an instrument. Assume that $H_{\min}(X|B)_\rho \geq k_1$ and $H_{\min}(Y|R)_{\cN[\sigma]} \geq k_2$ hold where $\sigma_{BR}$ is a purification of $\rho_{B}$.
    Let $\mathrm{Ext}$ be a $(k_1, k_2, \varepsilon)$ two-process extractor strong in $Y$.
    Then $\rho^{\mathrm{out}}_{XYT} \coloneqq \cN_{YT|B}[\rho_{XB}]$ is such that $Z = \mathrm{Ext}(X,Y)$ is $\varepsilon$-random relative to $YT$.
\end{lemma}
\begin{proof}
    Consider the state
    \begin{equation}
        \sigma_{BB'} \coloneqq \rho_{B}^{1/2}\Omega_{BB'}\rho_{B}^{1/2} = \left(\rho_{B'}^{1/2}\right)^T \Omega_{BB'} \left(\rho_{B'}^{1/2}\right)^T
    \end{equation}
    which is a purification of $\rho_B$. Define the channel
    \begin{equation}
        \cM_{X|B'}\left[ \sigma_{B'} \right] \coloneqq \tr_{B'}\left[ \left(\rho_{B'}^{- 1/2}\right)^T \rho_{XB'}^{T_{B'}} \left(\rho_{B'}^{- 1/2}\right)^T \sigma_{B'} \right].
    \end{equation}
    These then satisfy
    \begin{equation}
        \cM_{X|B'}\left[ \sigma_{BB'} \right] = \tr_{B'}\left[ \rho_{XB'}^{T_{B'}}\Omega_{BB'} \right] = \rho_{XB}.
    \end{equation}
    Hence
    \begin{equation}
        H_{\min}(X|B)_{\cM[\sigma]} = H_{\min}(X|B)_{\rho} \geq k_1
    \end{equation}
    and
    \begin{equation}
        \left(\cM_{X|B'} \otimes \cN_{YT|B}\right)[\sigma_{BB'}] = \rho^\mathrm{out}_{XYT}.
    \end{equation}
    Furthermore, by the isometric invariance of $H_{\min}$, we have
    \begin{equation}
        H_{\min}(Y|B')_{\cN[\sigma]} \geq k_2.
    \end{equation}
    Since $\mathrm{Ext}$ is a $(k_1, k_2, \varepsilon)$ two-process extractor strong in $Y$, we have that $Z = \mathrm{Ext}(X, Y)$ is $\varepsilon$-random relative to $YT$ as desired.
\end{proof}

Furthermore, as shown in the proof of \cref{thm:channel_ext}, any function that allows for extracting randomness from $\rho_{XB}$ and $\cN_{YT|B}$ as described by \cref{lem:alt_model_equivalence} is also a two-process extractor (with identical parameters).

\end{document}